\documentclass{sig-alternate}
\usepackage{float}
\usepackage{algorithm}
\usepackage{algorithmic}
\usepackage{graphicx}
\usepackage{cancel}
\usepackage[scriptsize,tight]{subfigure}
\usepackage{amsmath}
\usepackage{amssymb}
\usepackage{amsfonts}
\usepackage{cite}
\usepackage{url}
\usepackage{alltt}
\usepackage{colortbl}

\makeatletter
    \newcommand\figcaption{\def\@captype{figure}\caption}
    \newcommand\tabcaption{\def\@captype{table}\caption}
\makeatother

\newtheorem{definition}{Definition}[section]

\newtheorem{example}{Example}
\newtheorem{theorem}{Theorem}[section]

\newtheorem{corollary}{Corollary}[section]

\DeclareMathOperator*{\argmin}{arg}
\def\sharedaffiliation{%
\end{tabular}\newline \vspace{-6mm}\begin{tabular}{c}}
\begin{document}
\title{Towards Utility-driven Anonymization of Transactions}

\numberofauthors{3}
\author{
\alignauthor Grigorios Loukides\\
\alignauthor Aris Gkoulalas-Divanis\\
\alignauthor Bradley Malin\\
\sharedaffiliation
\affaddr{\email {\{grigorios.loukides, aris.gkoulalas, b.malin\}@vanderbilt.edu}}\\
\affaddr{Department of Biomedical Informatics }  \\
\affaddr{Vanderbilt University }   \\
\affaddr{Nashville, TN, USA }\\
}

\toappear{}
\maketitle

\begin{abstract}
Publishing person-specific transactions in an anonymous form is increasingly required by organizations. Recent approaches ensure that potentially identifying information (e.g., a set of diagnosis codes) cannot be used to link published transactions to persons' identities, but all are limited in application because they incorporate coarse privacy requirements (e.g., protecting a certain set of $m$ diagnosis codes requires protecting all $m$-sized sets), do not integrate utility requirements, and tend to explore a small portion of the solution space. In this paper, we propose a more general framework for anonymizing transactional data under specific privacy and utility requirements. We model such requirements as constraints, investigate how these constraints can be specified, and propose COAT (COnstraint-based Anonymization of Transactions), an algorithm that anonymizes transactions using a flexible hierarchy-free generalization scheme to meet the specified constraints. Experiments with benchmark datasets verify that COAT significantly outperforms the current state-of-the-art algorithm in terms of data utility, while being comparable in terms of efficiency. The effectiveness of our approach is also demonstrated in a real-world scenario, which requires disseminating a private, patient-specific transactional dataset in a way that preserves both privacy and utility in intended studies.
\end{abstract}

\section{Introduction}\label{introduction}

Organizations in various domains, ranging from healthcare to government, increasingly share person-specific data, devoid of explicit identifiers (e.g., names), to enable research and comply with regulations. For instance, the U.S. National Institutes of Health (NIH) recently mandated that NIH-sponsored investigators disclose data collected or studied in a manner that is ``free of identifiers that could lead to deductive disclosure of the identity of individual subjects" \cite{10111}. Numerous studies demonstrate that \emph{de-identification} (i.e., the removal of explicit identifiers) is insufficient for privacy protection of transactional data (i.e., data in which a set of \emph{items} correspond to an individual) \cite{145} \citen{142} \cite{LoukidesAMIA}. This is because published transactions can disclose the identity of an individual associated with a transaction, if an attacker knows some of the items this individual is associated with. Imagine for example that \emph{Alice} was diagnosed with the diseases contained in the first transaction of Fig. \ref{original_data} and told her neighbor \emph{Bob} that she suffers from $a,b$ and $c$, which are relatively common. Publishing a de-identified version of the data of Fig. \ref{original_data} allows \emph{Bob} to find out that the first transaction corresponds to \emph{Alice}, since there is only one transaction containing $a, b$ and $c$ in this dataset. This problem, referred to as \emph{identity disclosure}, must be addressed to comply with regulations and to protect individuals' privacy. Having identified \emph{Alice}'s transaction, for example, \emph{Bob} can infer that \emph{Alice} also suffers from the diseases $d,e,f,g$ and $h$.

\subsection{Motivation}\label{motivation_intro}

To prevent identity disclosure, portions of transactions that are potentially linkable to identifying information need to be explicitly specified and protected prior to data release. This involves formulating a set of \emph{privacy constraints}, which are satisfied by transforming data so that each individual is linked to a sufficiently large number of transactions with respect to these constraints. This process achieves privacy because an attacker must distinguish an individual's real transaction among the transformed ones to identify him/her. However, data transformation may harm data utility when usability requirements are unaccounted for, resulting in released data that is subpar for intended applications. Thus, it is essential to balance privacy constraints with \emph{utility constraints} to ensure meaningful analysis. To show the importance of both types of constraints consider Example \ref{example2}.

\begin{figure}[!ht]
\begin{center}
\begin{minipage}[u]{1in}
\subfigcapskip=7pt
\parbox{1.4in}{
\subfigure[Original dataset]
{\small
\begin{tabular}{|r|l|}
\hline Patient & \emph{Diagnosis Codes} \tabularnewline
\hline \hline \emph{Alice} & a, b, c, d, e, f, g, h \tabularnewline \hline
 \emph{Mary} & a, c, e, f, g \tabularnewline \hline
 \emph{George} & c, d, e, f, h \tabularnewline\hline
 \emph{Jack} & a, c, e, f \tabularnewline \hline
 \emph{Anne} & e, f, g, h \tabularnewline\hline
 \emph{Tom} & d, e, f, g \tabularnewline\hline
 \emph{Jim} & a, b, d, e \tabularnewline\hline
 \emph{Steve} & a, c, f \tabularnewline\hline
 \emph{David} & a, c \tabularnewline\hline
 \emph{Ellen} & b, h \tabularnewline\hline
\end{tabular}\label{original_data}
}}
\end{minipage}\hspace{+22mm}
\begin{minipage}[u]{0.88in}
\subfigcapskip=7pt
\begin{minipage}[u]{1.4in}
\vspace{+6mm}
\subfigcapskip=7pt
\subfigcapmargin=-10pt
\subfigure[Privacy Constraints]
{\small
\begin{tabular}{|c|}
\hline
  $p_1$=\{a,b,c\} \tabularnewline\hline
  $p_2$=\{d,e,f,g,h\} \tabularnewline \hline
\end{tabular}\label{original_data_pc}
}
\end{minipage}
\hspace{+1mm}
\begin{minipage}[u]{1.4in}
\vspace{+5.2mm}
\subfigcapskip=7pt
\subfigcapmargin=-10pt
\subfigure[Utility Constraints]
{\small
\begin{tabular}{|c|}
\hline
  $u_1$=\{a,b\} \tabularnewline\hline
  $u_2$=\{c\}\tabularnewline\hline
  $u_3$=\{d\}\tabularnewline\hline
  $u_4$=\{e,f,g,h\} \tabularnewline\hline
\end{tabular}\label{original_data_uc}
}\end{minipage}
\end{minipage}

\vspace{-3mm}
\caption{Example dataset and constraints.}\label{original_data_all}
\end{center}
\end{figure}\vspace{-6mm}

\subfigcapskip=7pt
\begin{figure*}
\begin{center}
\subfigcapskip=7pt
\begin{minipage}[b]{1.55in}
\subfigure[Apriori ($5^5$-anonymity) \cite{140}]
{\small\centering
\begin{tabular}{>{\columncolor[rgb]{0.8,0.8,0.75}}r||l|}
\cline{2-2} {\bf Patient} & \emph{Diagnosis Codes} \tabularnewline
\cline{2-2} \emph{Alice} & $(a,b,c)$, $(d,e,f,g,h)$ \tabularnewline \cline{2-2}
 \emph{Mary} & $(a,b,c)$, $(d,e,f,g,h)$ \tabularnewline \cline{2-2}
 \emph{George} & $(a,b,c)$, $(d,e,f,g,h)$ \tabularnewline\cline{2-2}
 \emph{Jack} & $(a,b,c)$, $(d,e,f,g,h)$ \tabularnewline \cline{2-2}
 \emph{Anne} & $(d,e,f,g,h)$ \tabularnewline\cline{2-2}
 \emph{Tom} & $(d,e,f,g,h)$ \tabularnewline\cline{2-2}
 \emph{Jim} & $(a,b,c)$, $(d,e,f,g,h)$ \tabularnewline\cline{2-2}
 \emph{Steve} & $(a,b,c)$, $(d,e,f,g,h)$ \tabularnewline\cline{2-2}
 \emph{David} & $(a,b,c)$ \tabularnewline\cline{2-2}
 \emph{Ellen} & $(a,b,c)$, $(d,e,f,g,h)$ \tabularnewline\cline{2-2}
\end{tabular}\label{anon_a}
}
\end{minipage}\hspace{+13mm}
\begin{minipage}[b]{1.6in}
\subfigcapmargin=14pt
\subfigure[Greedy Algorithm \mbox{($(1,5,5)$-coherence) \cite{141}}]
{\small\centering
\begin{tabular}{>{\columncolor[rgb]{0.8,0.8,0.75}}r||l|}
\cline{2-2} {\bf Patient} & \emph{Diagnosis Codes} \tabularnewline
\cline{2-2} \emph{Alice} & e, f \tabularnewline \cline{2-2}
 \emph{Mary} & e, f \tabularnewline \cline{2-2}
 \emph{George} & e, f \tabularnewline\cline{2-2}
 \emph{Jack} & e, f \tabularnewline \cline{2-2}
 \emph{Anne} & e, f \tabularnewline\cline{2-2}
 \emph{Tom} & e, f \tabularnewline\cline{2-2}
 \emph{Jim} & e \tabularnewline\cline{2-2}
 \emph{Steve} & f \tabularnewline\cline{2-2}
 \emph{David} & - \tabularnewline\cline{2-2}
 \emph{Ellen} & - \tabularnewline\cline{2-2}
\end{tabular}\label{anon_b}
}\end{minipage}\hspace{+13mm}
\begin{minipage}[b]{1.8in}
\subfigure[COAT ($k=5, s=15\%$, constraints of Figs. \ref{original_data_pc} and \ref{original_data_uc})]
{\small\centering
\begin{tabular}{>{\columncolor[rgb]{0.8,0.8,0.75}}r||l|}
\cline{2-2} {\bf Patient} & \emph{Diagnosis Codes} \tabularnewline
\cline{2-2} \emph{Alice} & $(a,b)$, c, e, f, $(g,h)$ \tabularnewline \cline{2-2}
 \emph{Mary} & $(a,b)$, c, e, f, $(g,h)$ \tabularnewline \cline{2-2}
 \emph{George} & ~c, e, f, $(g,h)$ \tabularnewline\cline{2-2}
 \emph{Jack} & $(a,b)$, c, e, f \tabularnewline \cline{2-2}
 \emph{Anne} & ~e, f, $(g,h)$ \tabularnewline\cline{2-2}
 \emph{Tom} & ~e, f, $(g,h)$ \tabularnewline\cline{2-2}
 \emph{Jim} & $(a,b)$, e \tabularnewline\cline{2-2}
 \emph{Steve} & $(a,b)$, c, f \tabularnewline\cline{2-2}
 \emph{David} & $(a,b)$, c \tabularnewline\cline{2-2}
 \emph{Ellen} & $(a,b)$, $(g,h)$ \tabularnewline\cline{2-2}
\end{tabular}\label{an_our}
}
\end{minipage}
\end{center}
\vspace{-5mm}
\caption{Three possible anonymizations of the dataset of Fig. \ref{original_data}.}
\end{figure*}

\begin{example}
Imagine that a hospital needs to publish the dataset of Fig. \ref{original_data}, where each transaction corresponds to a patient (patients' names are not released) and consists of a set of diagnosis codes (\textbf{items}). Certain item combinations (\textbf{itemsets}), such as the combinations of diagnosis codes $abc$ and $defgh$ of Fig. \ref{original_data_pc}, are regarded as potentially linkable and must be associated with at least $5$ transactions to prevent identity disclosure. At the same time, the published data must be able to support a study in which the number of patients diagnosed with \emph{cold}, denoted with $c$, needs to be accurately determined. These requirements can be modeled via the privacy constraints in Fig. \ref{original_data_pc} and a utility constraint $\{c\}$. An anonymization that satisfies them is given in Fig. \ref{an_our}. Observe
that each patient can be linked to no less than $5$ transactions using the combinations $abc$ or $defgh$, while the number of patients suffering from \emph{cold} can still be accurately computed after anonymization.
\label{example2}
\end{example}

\subsection{Limitations of Existing Methodologies}

Preventing identity disclosure in transactional data was investigated recently \cite{140,142}; however, existing approaches inadequately deal with the scenario of Example \ref{example2} for several reasons. First, they support a limited class of privacy requirements. While they are effective at protecting all itemsets comprised of a certain number of items (i.e., itemsets of certain {\bf size}),  potentially linkable itemsets may involve certain items only and vary in size. For instance, in the presence of the privacy constraints in Example \ref{example2}, the approaches of \cite{140} and \cite{142} would protect all $56$ combinations of $5$ diagnosis codes (e.g., $abcde, bcdef$, etc.). Unnecessarily protecting itemsets significantly distorts data because the number of itemsets that require protection rapidly increases with their size.

Second, prior approaches neglect specific data utility requirements. Thus, they do not guarantee generating practically useful solutions for environments where usability is based on well-defined policies, such as in epidemiology (where combinations of diagnosis codes form syndromes) \cite{157}. For instance, when applied to the dataset of Fig. \ref{original_data}, \cite{140} and \cite{142} produce the anonymizations of Figs. \ref{anon_a} and \ref{anon_b} (respectively). These anonymizations do not support the study of Example 1 because they do not allow the number of patients suffering from \emph{cold} to be accurately computed, as a result of violating the utility constraint $\{c\}$.

Third, the existing literature considers only a small number of possible transformations to meet privacy constraints. For example, \cite{140} protects $a, b$ and $c$ by generalizing them to their closest common ascendant $(a,b,c)$ according to the hierarchy of Fig. \ref{hier1}, while \cite{142} does so by eliminating them from the released dataset. This incurs excessive information loss, but, as we show in this paper, is unnecessary because items can be generalized together in a hierarchy-free manner to reduce the amount of information loss incurred.

\begin{figure}[!ht]
\small\centering
\begin{tabular}{c}
\includegraphics[scale=0.5]{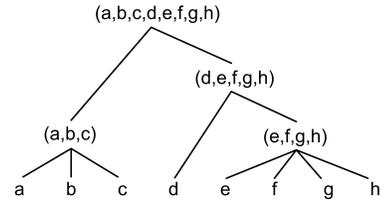}
\end{tabular}\vspace{-3mm}
\caption{Hierarchy for the dataset of Fig. \ref{original_data}.}\label{hier1}
\end{figure}
\subsection{Contributions}\label{contributions_intro}

This paper proposes an innovative approach to anonymize transactional data under privacy and utility constraints. Given a dataset and a set of constraints, our approach prevents identity disclosure by ensuring that each transaction is indistinguishable from at least $k-1$ other transactions with respect to privacy constraints, while satisfying utility constraints. For instance, when applied to anonymize the data in Fig. \ref{original_data} using the constraints in
Figs. \ref{original_data_pc} and \ref{original_data_uc}, our approach generates the anonymization of Fig. \ref{an_our}, which satisfies the imposed constraints.

\subsubsection*{Our work makes the following specific contributions:}

First, we propose a novel constraint specification model that enables data owners to express detailed privacy and utility requirements. We introduce privacy constraints that allow protecting the required itemsets only, thereby reducing information loss, and utility constraints that ensure the utility of anonymized data in practice. Acknowledging that constraint specification may be difficult in lack of specific domain knowledge, we also propose an algorithm to extract privacy constraints from the data and a recipe for specifying utility constraints.

Second, we propose an item generalization model that eliminates the requirement for hierarchies. This is important because many domains do not fit into rigid hierarchies \cite{Hepp,R2006}. In doing so, our framework can produce fine-grained anonymizations with substantially better utility than that of \cite{140} and \cite{142}. For example, to meet the privacy constraint $p_1$ (Fig. \ref{original_data_pc}), our solution generalizes $a$ and $b$ to $(a, b)$ and leaves $c$ intact, retaining more information than \cite{140}, which releases $(a, b, c)$, and \cite{142}, which eliminates $a, b$ and $c$.

Third, we develop COnstraint-based Anonymization of Transactions (COAT), an algorithm that iteratively selects privacy constraints and transforms data to satisfy them.
For each privacy constraint, COAT generalizes items in accordance with the specified utility constraints and attempts to minimize information loss. When a privacy constraint cannot be satisfied through generalization, COAT suppresses the least number of items required to meet this constraint.

Fourth, we investigate the effectiveness of our approach through experiments on widely-used benchmark datasets and a case study on patient-specific data extracted from the Electronic Medical Record system \cite{Stead} of Vanderbilt University Medical Center, a large healthcare provider in the U.S. Our results verify that the proposed methodology is able to anonymize transactions under various privacy and utility constraints with less information loss than the state-of-the-art method \cite{140} and to generate high-quality anonymizations in a real-world data dissemination scenario.

\subsection{Paper Organization}\label{paper_org}

The rest of the paper is organized as follows. Related work is reviewed in Section \ref{related_work}. We formally define our constraint specification model and the problem of anonymizing transactions under constraints in Section \ref{bp_def}. COAT is presented in Section \ref{algorithm}. Section \ref{up_specification} presents methods for specifiying constraints. Sections \ref{experiments} and \ref{casestudy} report experimental results and a case study of applying COAT respectively. Finally, we conclude the paper in Section \ref{conclusion}.

\section{Related Work}\label{related_work}
The problem of identity disclosure in \emph{relational} data publishing has been studied extensively \cite{32,23,20,8,34}. A well-established principle that can prevent this type of threat is $k$-anonymity \cite{32,23}. A relational table is $k$-anonymous when each record is indistinguishable from at least $k-1$ others with respect to potentially identifying attributes (termed quasi-identifiers or \emph{QID}s). $K$-anonymity can be achieved by \emph{generalization}, a process in which \emph{QID} values are replaced by more general ones specified by a generalization (recoding) model, or via \emph{suppression}, a technique that removes values or records from anonymized data \cite{20}. In this work, we anonymize transactions to thwart identity disclosure by applying item generalization and suppression. Beyond identity disclosure is another threat in relational data publishing known as \emph{attribute disclosure} (i.e., the inference of an individual's sensitive values), which can be guarded against by several principles, such as $l$-diversity \cite{2}. We note that our approach can be extended to prevent attribute disclosure as well, but this is beyond the scope of this paper.

Privacy-preserving publication of transactions was recently investigated \cite{140,141,142}. First, in \cite{140}, $k^m$-anonymity was proposed to prevent attackers with the knowledge of at most $m$ items from linking an identified individual to less than $k$ published transactions. The authors of \cite{140} designed three algorithms to enforce $k^m$-anonymity, but the Apriori algorithm is the only one that is sufficiently scalable for use in practice. It operates in a bottom-up fashion, beginning with itemsets comprised of one item and subsequently considers incrementally larger itemsets. In each iteration, $k^m$-anonymity is enforced using a hierarchy-based generalization model. Second, \cite{142} proposed $(h,k,p)$-coherence, a privacy principle that addresses both identity and attribute disclosure. This principle assumes a fixed classification of items into potentially linkable and sensitive, treats potentially linkable items similarly to $k^m$-anonymity, and additionally limits the probability of inferring sensitive items. Following \cite{140}, we do not adopt such a classification, and allow any item to be treated as potentially linkable for the remaining (sensitive) items. To satisfy $(h,k,p)$-coherence, \cite{142} proposed an algorithm that discovers all unprotected itemsets of minimal size and protects them by iteratively suppressing the item contained in the greatest number of those itemsets. The primary differences between our work and the approaches of \cite{140} and \cite{142} were discussed in the Introduction. Finally, \cite{141} developed a method that eliminates attribute disclosure based on \emph{bucketization} \cite{79} and $l$-diversity. Our work is orthogonal to \cite{141}; we do not aim to thwart attribute disclosure, but rather apply item generalization and suppression to guard against identity disclosure.

Preserving privacy has also been considered in contexts related to knowledge sharing, where the goal is to prevent the inference of rules or patterns \cite{156,144}. Our method is fundamentally different from this line of research, as we aim to publish data that prevents the disclosure of individuals' identities instead.

\section{Background and Problem Formulation}\label{bp_def}
Let $\mathcal{I}=\{i_1,...,i_M\}$ be a finite set of literals, called \emph{items}. Any subset $I\subseteq\mathcal{I}$ is called an \emph{itemset} over $\mathcal{I}$, and is represented as the concatenation of the items it contains. An itemset that has $m$ items or equivalently a \emph{size} of $m$, is called an $m$-itemset. A dataset $\mathcal{D}=\{T_1,...,T_N\}$ is a set of $N$ transactions. Each \emph{transaction} $T_n$, $n=1,...,N$, over $\mathcal{I}$ corresponds to a unique individual and is a pair $T_n=(\mbox{\textit{tid}},I)$, where $I$ is the itemset and $\mbox{\textit{tid}}$ is a unique identifier. A transaction $T_n=(\mbox{\textit{tid}},J)$ \emph{supports} an itemset $I$ over $\mathcal{I}$, if $I \subseteq J$. Given an itemset $I$ over $\mathcal{I}$ in $\mathcal{D}$, we use $sup(I, \mathcal{D})$ to represent the number of transactions $T_n\in \mathcal{D}$ that support $I$. This set of transactions, called \emph{supporting transactions} of $I$ in $\mathcal{D}$, is denoted as $\mathcal{D}_I$.

\subsection{Set-based Anonymization}\label{generalization_model}

We propose a set-based anonymization model for transactional data, formally defined as follows:

\begin{definition}\hspace{+1mm}\textsc{(Set-Based Anonymization).}
A set-based anonymization of $\mathcal{I}$ is a set $\tilde{\mathcal{I}}=\{\tilde{i_1},...,\tilde{i}_{\tilde{M}}\}$
with the following properties:
(1) each item in $\mathcal{I}$ is mapped to a unique item $\tilde{i_m}\in \tilde{\mathcal{I}}$, $m \in [1, \tilde{M}]$, that is a subset of $\mathcal{I}$, using an anonymization function $\Phi: \mathcal{I} \rightarrow \tilde{\mathcal{I}}$, (2) $\bigcup_{m=1}^{\tilde{M}}\tilde{i_m}=\mathcal{I}-\mathcal{S}$, where $\mathcal{S}$ is the set of items mapped to the empty subset of $\mathcal{I}$, and (3) $\tilde{i_r}\cap\tilde{i_s}=\emptyset$, for any $\tilde{i_r},\tilde{i_s} \in \tilde{\mathcal{I}}$, $r\neq s$.
\end{definition}

$\tilde{i}$ is a \emph{generalized item} when it contains at least one $i_r \in \mathcal{I}$ that is mapped to a non-empty subset of $\mathcal{I}$. We use the notation $\tilde{i} = (i_1,\ldots,i_m)$ to refer to its elements (items) from $\mathcal{I}$. Any item from $\mathcal{I}$ that is mapped to the empty subset of $\mathcal{I}$, denoted as $(~)$, is called \emph{suppressed}, and is contained in the set $\mathcal{S}$. An example of a set-based anonymization is shown in Fig. \ref{mapf}. Notice that  items $a$ and $b$ are mapped to the same generalized item $(a,b)$, whereas item $d$ is suppressed.

\begin{figure}[!ht]
\scriptsize\centering
\includegraphics[scale=0.55]{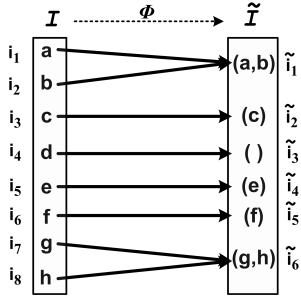}
\vspace{-3mm}
\caption{An example of set-based anonymization.}\label{mapf}
\end{figure}

The set-based anonymization model is very flexible because it does not force any items to be generalized together, as formally shown in Corollary \ref{corollary:no-force}. This is different from the \emph{full-subtree} generalization model \cite{16} adopted in \cite{140}, which forces all siblings of an original (leaf-level) item to be mapped to an intermediate node in the hierarchy when this item is generalized to the intermediate node.

\begin{corollary}
In the set-based anonymization model, mapping an item $i_r \in \mathcal{I}$ to a generalized item $\tilde{i}$ does not force any other item $i_s \in \mathcal{I}$ to be mapped to $\tilde{i}$.
\label{corollary:no-force}\end{corollary}

Additionally, as explained in Corollary \ref{corollary:special-case}, the set-based anonymization model contains the generalization model used in \cite{140} as a special case. Thus, our model enables exploring a much larger set of possible anonymizations, which have the potential to incur less information loss.

\begin{corollary}
The \emph{full-subtree} recoding model is a special case of the set-based anonymization model, where each $\tilde{i_m}$, $m \in [1, \tilde{M}]$, is mapped to an intermediate node of the considered hierarchy.\label{corollary:special-case}
\end{corollary}

Our anonymization model transforms a dataset $\mathcal{D}$ into a new dataset $\tilde{\mathcal{D}}$ that helps prevent identity disclosure, since the number of transactions of $\tilde{\mathcal{D}}$ that can be associated with an individual is increased, as proven in Theorem \ref{theorem:gen-principle}.

\begin{theorem}\label{theorem:gen-principle}\hspace{+1mm}\textsc{(Generalization Principle).}
Given two items $i_r, i_s$ that appear in transactions of $\mathcal{D}$ and are mapped to the same generalized item $\tilde{i}$ after anonymizing $\mathcal{D}$ to $\tilde{\mathcal{D}}$, and an itemset $i_ri_s$, it holds that $$sup(\tilde{i}, \tilde{\mathcal{D}})=sup(i_r,\mathcal{D})+sup(i_s,\mathcal{D})-sup(i_ri_s,\mathcal{D})$$
\end{theorem}

\begin{proof}
The proof follows directly from the fact that the items $i_r$ and $i_s$, and the itemset $i_ri_s$ are mapped to a common literal $\tilde{i} \in \tilde{\mathcal{D}}$ in all transactions of $\mathcal{D}$ that support $i_r, i_s$ or $i_ri_s$.
\end{proof}

We illustrate Theorem \ref{theorem:gen-principle} using Figs. \ref{original_data} and \ref{an_our}. Consider items $a$, $b$ and the itemset $ab$ in Fig. \ref{original_data}, which have support of $6$, $3$ and $2$ respectively, and are mapped to the same generalized item $(a,b)$ in Fig. \ref{an_our}. Observe that $(a,b)$ has a support of $7$ that is equal to the sum of the supports of $a$, $b$ minus that of $ab$.

\subsection{Privacy Constraints}\label{privacy_constraint_subsection}

The integration of privacy constraints is central to our framework because they allow for the explicit definition of which itemsets are potentially linkable and require protection. In what follows, we formally define the notion of privacy constraints and their satisfiability.

\begin{definition}\label{pc_set_def}\hspace{+1mm}\textsc{(Privacy Constraint Set).}
A privacy constraint $p$ is a non-empty set of items in $\mathcal{I}$ that are specified as potentially linkable. The union of all privacy constraints formulates a privacy constraint set $\mathcal{P}$.
\end{definition}

\begin{definition}\label{pc_satisfaction}\hspace{+1mm}
\textsc{(Privacy Constraint Satisfiability).}
Given a parameter $k$, a privacy constraint $p = \{i_1,...,i_r\} \in \mathcal{P}$ is \emph{satisfied} when the corresponding itemset $\bigcup_{m=1}^{r}\Phi(i_m)$ is: (1) supported by at least $k$ transactions in $\tilde{\mathcal{D}}$, or (2) not supported in $\tilde{\mathcal{D}}$ and
each of its proper subsets is either supported by at least $k$ transactions in $\tilde{\mathcal{D}}$ or not supported in $\tilde{\mathcal{D}}$. $\mathcal{P}$ is satisfied when every $p \in \mathcal{P}$ is satisfied.
\end{definition}

To illustrate these definitions, consider the privacy constraint $p_1=\{a, b, c\}$ in Fig. \ref{original_data_pc}. This privacy constraint is satisfied for $k=5$ in the dataset of Fig. \ref{an_our} because $5$ transactions support the itemset $\Phi(a)\cup\Phi(b)\cup\Phi(c)=(a,b)c$.

Satisfying a privacy constraint $p$ prevents identity disclosure because the number of transactions that can be linked to an individual using any subset of items in $p$ is either at least $k$, or zero, as shown in Theorem \ref{quality_constraints_monotonicity}.

\begin{theorem}\label{quality_constraints_monotonicity}\hspace{+1mm}\textsc{(Monotonicity).}
For a given $k$, the satisfaction of a privacy constraint $p$ in $\tilde{\mathcal{D}}$ implies that each privacy constraint $p_j\subset p$ is satisfied in $\tilde{\mathcal{D}}$.
\end{theorem}
\begin{proof}
Assume that a privacy constraint $p_j\subset p$ is not satisfied in $\tilde{\mathcal{D}}$ for this value of $k$. Then, according to Definition \ref{pc_satisfaction}, the satisfaction of $p$ implies that the itemset $I=\bigcup_{\forall i_m\in p}\Phi(i_m)$ is supported by either at least $k$ or $0$ transactions in $\tilde{\mathcal{D}}$. Now, consider an itemset $J=\bigcup_{\forall i_m \in p_j}\Phi(i_m)$, which is derived by applying $\Phi$ on each item in $p_j$. Since $J\subset I$, we have $sup(J,\tilde{\mathcal{D}}) > sup(I,\tilde{\mathcal{D}})$ when $sup(I,\tilde{\mathcal{D}}) \geq k$ due to the monotonicity principle \cite{150}. When $sup(I,\tilde{\mathcal{D}})=0$, $p_j$ is satisfied by Definition \ref{pc_satisfaction}. In either case $p_j$ is satisfied in $\tilde{\mathcal{D}}$ for the given $k$, which contradicts the assumption and proves that the theorem holds true.
\end{proof}

Our privacy constraint specification model offers two benefits. First, it allows data owners to specify a range of different privacy requirements. For instance, it can be used to protect specific itemsets of various sizes, or to provide the same privacy guarantees as $k^m$-anonymity (by formulating a privacy constraint set that consists of all itemsets of size $m$). Second, our model allows protecting any set of itemsets without enforcing any additional itemsets to be unnecessarily protected. This is important because unecessarily protecting itemsets may significantly increase the amount of information loss incurred to anonymize data.

\subsection{Utility Constraints} \label{utility_constraints_section}

Privacy protection is offered at the expense of data utility \cite{32,20}, and so it is important to ensure that anonymized data is not overly distorted. Existing approaches attempt to do so by minimizing the amount of information loss incurred when anonymizing transactions \cite{140,141}, but do not guarantee furnishing a useful result for intended applications. By contrast, our methodology offers such guarantees
through the introduction of utility constraints. Before formally defining such constraints,
we make the following important observations related to data usefulness.
\centerline{}
\noindent {\bf Observation 1} Mapping a set of items in $\mathcal{D}$ to the same generalized item in the anonymized dataset $\tilde{\mathcal{D}}$ introduces distortion because these items become indistinguishable in $\tilde{\mathcal{D}}$. When there is no control of how specific items are generalized, $\tilde{\mathcal{D}}$ may not be practically useful.
\centerline{}
\noindent {\bf Observation 2} Suppressing an item in $\mathcal{D}$ introduces distortion because this item is not contained in $\tilde{\mathcal{D}}$ and the amount of distortion increases with the number of suppressions.
\centerline{}

Based on these observations, we introduce a utility constraint set $\mathcal{U}$ to limit the amount of generalization items are allowed to receive based on application requirements, and bound the number of items that can be suppressed using a threshold $s$. Definitions \ref{utility_constraint_set_def} and \ref{uc_satisfaction} illustrate the definition of a utility constraint set and its satisfiability respectively.

\begin{definition}\label{utility_constraint_set_def}\hspace{+1mm}\textsc{(Utility Constraint Set).}
A utility constraint set $\mathcal{U}$ is a partition of $\mathcal{I}$ that declares the set of allowable mappings of the items from $\mathcal{I}$ to those of $\tilde{\mathcal{I}}$ through $\Phi$. Each element of $\mathcal{U}$ is called a utility constraint.
\end{definition}

\begin{definition}\label{uc_satisfaction}\hspace{+1mm}\textsc{(Utility Constraint Set Satisfiability).}
Given sets $\mathcal{I}$, $\tilde{\mathcal{I}}$, a utility constraint set $\mathcal{U}$, and a parameter $s$, $\mathcal{U}$ is \emph{satisfied}
if and only if (1) for each non-empty $\tilde{i_m} \in \tilde{\mathcal{I}}$, $\exists u_j \in \mathcal{U}$ such that $\tilde{i_m} \subseteq u_j$, and
(2) the fraction of items in $\mathcal{I}$ contained in the set of suppressed items $\mathcal{S}$ is at most $s\%$.
\end{definition}

The first condition limits the maximum amount of generalization each item is allowed to receive in a set-based anonymization $\mathcal{\tilde{I}}$, while the second condition ensures that the number of suppressed items is controlled by a threshold specified by data owners. When both of these conditions hold, $\mathcal{U}$ is satisfied, and $\mathcal{\tilde{I}}$ corresponds to a dataset that can be meaningfully analyzed. Example \ref{last_ex} illustrates the above definitions.

\begin{example}
Consider Fig. \ref{util_map}, in which $\tilde{i_1}\subseteq u_1$, $\tilde{i_2}\subseteq u_2$, and $\tilde{i_4},\tilde{i_5},\tilde{i_6}$ are subsets of $u_4$.
The percent of suppressed items is $12.5\%$ because $\mathcal{I}$ consists of $8$ items (see Fig. \ref{mapf}), and only $d$ is suppressed. Thus, $\mathcal{U}$ is satisfied. \label{last_ex}
\end{example}
\vspace{-3mm}
\begin{figure}[!ht]
\scriptsize\centering
\includegraphics[width=0.35\columnwidth]{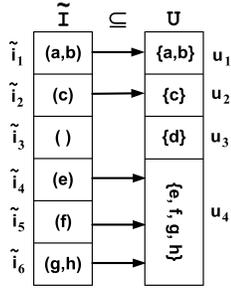}
\vspace{-3mm}
\caption{Utility Constraint Set Satisfiability example for the set-based anonymization of Fig. \ref{mapf}}\label{util_map}
\end{figure}

We also observe that the number of the supporting transactions of a \emph{generalized} item in the anonymized dataset $\tilde{\mathcal{D}}$ is equal to the number of transactions supporting \emph{any} item in the original dataset $\mathcal{D}$ that is mapped to this generalized item, as illustrated in Theorem \ref{mytheorem}.

\begin{theorem}\label{mytheorem}
Given a generalized item $\tilde{i_m}\in \tilde{\mathcal{I}}$ such that $\tilde{i_m}=\{i_1,...,i_r\}$, it holds that
$$|\mathcal{\tilde{D}}_{\tilde{i_m}}|=|\mathcal{D}_{i_1}\cup\ldots\cup\mathcal{D}_{i_r}|$$
\noindent where $|\mathcal{\tilde{D}}_{\tilde{i_m}}|$ denotes the size of the set of supporting transactions of
$\tilde{i_m}$ in $\mathcal{\tilde{D}}$, and $|\mathcal{D}_{i_1}\cup\ldots\cup\mathcal{D}_{i_r}|$ the size of the set
of transactions supporting at least one of the items in $\{i_1,...,i_r\}$ in $\mathcal{D}$.
\end{theorem}
\begin{proof}
The proof is omitted, because it is similar to that of Theorem \ref{theorem:gen-principle}.
\end{proof}

Based on Theorem \ref{mytheorem}, we provide the following corollary that highlights the importance of anonymizing data while satisfying utility constraints.
\begin{corollary}\label{uc_coro}
Given a utility constraint set $\mathcal{U}$ that is satisfied, a utility constraint $u_j=\{i_1,...,i_r\} \in \mathcal{U}$, and a set of generalized items $\{\tilde{i_1},...,\tilde{i_s}\}$
constructed by mapping each element of $u_j$ to one of these items, it holds that
$$|\mathcal{\tilde{D}}_{\tilde{i_1}}\cup\ldots\cup\mathcal{\tilde{D}}_{\tilde{i_s}}|=|\mathcal{D}_{i_1}\cup\ldots \cup \mathcal{D}_{i_r}|$$
\end{corollary}

Thus, the number of transactions of $\mathcal{D}$ supporting any item contained in a utility constraint $u_j\in \mathcal{U}$ can be accurately computed from the anonymized dataset $\mathcal{\tilde{D}}$, when $\mathcal{U}$ is satisfied and all items in $u_j$ have been generalized. This is crucial in many data analysis tasks (e.g., in generalized association rule mining \cite{Srikant}) where the support of itemsets corresponding to \emph{aggregate concepts} (i.e., itemsets with a more general meaning than the items they are comprised of) needs to be determined, as illustrated below.

\begin{example}\label{diabetesexample}
Consider that the dataset of Fig. \ref{original_data} has to be anonymized to support a study in which the number of patients diagnosed with
\emph{diabetes} needs to be accurately computed. Assume also that diagnosis codes $a$ and $b$ correspond to two different forms of \emph{diabetes}. Observe that the number of patients suffering from \emph{diabetes} (i.e., transactions having $a$, $b$ or $ab$) in the dataset of Fig. \ref{original_data} is the same as in the anonymization of this dataset shown in Fig. \ref{an_our}, because this anonymization satisfies the utility constraints of Fig. \ref{original_data_uc} and both $a$ and $b$ in $u_1=\{a,b\}$ have been generalized.
\end{example}

\subsection{Information Loss}\label{information_loss_section}

There may be many anonymizations that satisfy the privacy and utility constraint sets, but they may not be equally useful. Since discovering the one that least harms data utility is important, we propose a measure to capture data utility based on information loss.

\begin{definition}\label{UL_item}\hspace{+1mm}\textsc{(Utility Loss for a Generalized Item).}
The Utility Loss (\textit{UL}) for a generalized item $\tilde{i_m}$ is defined as
\vspace{-2mm}
$$\mbox{\textit{UL}}(\tilde{i_m})=\frac{2^{|\tilde{i_m}|}-1}{2^{M}-1}\times w(\tilde{i_m})\times \frac{sup(\tilde{i_m},\tilde{\mathcal{D}})}{N}$$

\noindent where $|\tilde{i_m}|$ denotes the number of items from $\mathcal{I}$ mapped to $\tilde{i_m}$ using $\Phi$, and $w: \tilde{\mathcal{I}}\rightarrow [0,1]$ is a function assigning a weight to $\tilde{i_m}$.
\end{definition}

$\mbox{\textit{UL}}$ measures the amount of information loss caused by generalizing a set of items as a product of three terms. The first term penalizes a generalized item based on the number of items from $\mathcal{I}$ mapped to it. This is because a generalized item can be interpreted as any of the $2^{|\tilde{i_m}|}-1$ non-empty subsets of the set of items mapped to it \cite{141}, and there are $2^{M}-1$ possible non-empty subsets that can be formed using items from $\mathcal{I}$. The second term is a weight specified by data owners to quantify the harm to data utility caused by a generalized item, according to the items mapped to it. Weights need to be between $0$ and $1$ for normalization purposes, where larger weights are assigned to generalized items comprised of items that are more semantically distant, since such generalized items harm data utility more \cite{140}. The semantic distance of items can be computed in many ways (e.g., based on the height of a hierarchy \cite{23}, the number of leaves of the closest common ascendant of these items in a hierarchy \cite{14}, with the aid of ontologies \cite{155}, or by expert knowledge). The third term is the support of a generalized item in the anonymized dataset, normalized by the number of transactions. Items that appear often in the anonymized dataset are penalized more, since they introduce more data distortion. Example \ref{ulcomputation} illustrates how $\mbox{\textit{UL}}$ can be computed.

\begin{example}\label{ulcomputation}
Consider Fig.~\ref{original_data}, and the anonymized dataset of Fig. \ref{an_our}. Items $a$ and $b$ are generalized to $(a,b)$, which is assigned
a weight of $0.375$ specified by the data owner, and has a support of $7$ in Fig. \ref{an_our}. The $\mbox{\textit{UL}}$ for $(a,b)$ is computed as
$\frac{2^{2}-1}{2^{8}-1} \times 0.375 \times \frac{7}{8} \approx 0.004$.
\end{example}

Based on Definition \ref{UL_item}, we quantify the total amount of information loss for an anonymized dataset $\tilde{\mathcal{D}}$ as follows.

\begin{definition}\label{UL_dataset}\hspace{+1mm}\textsc{(Utility Loss for an Anonymized Dataset).}
The Utility Loss for an anonymized dataset $\tilde{\mathcal{D}}$ is given by
$$\mbox{\textit{UL}}(\tilde{\mathcal{D}})=\displaystyle\sum_{\forall \tilde{i_m} \neq \emptyset} \mbox{\textit{UL}}(\tilde{i_m})+\displaystyle\sum_{\forall i_m \in \mathcal{S}}Y(i_m)$$
\noindent where $Y:\mathcal{I}\rightarrow \Re$ is a function assigning a penalty to each suppressed item $i_m$ from $\mathcal{D}$.
\end{definition}

The above definition captures data utility loss caused by both generalization and suppression. Specifically, for suppression, similar to \cite{142}, we
allow data owners to assign a penalty to each suppressed item, according to the perceived importance of retaining this item in the anonymized result.
For instance, each suppressed item could receive a penalty equal to its support, based on the fact that this defines the number of transactions from which it is eliminated.

\subsection{Problem Statement}

\noindent {\it Given a transactional dataset $\mathcal{D}$, a privacy constraint set $\mathcal{P}$, a utility constraint set $\mathcal{U}$, and parameters $k, s$,
construct an anonymized version $\tilde{\mathcal{D}}$ of $\mathcal{D}$ using the set-based anonymization model such that: (1) $\mathcal{P}$ and $\mathcal{U}$ are both satisfied, and (2) the amount of utility loss $\mbox{\textit{UL}}(\tilde{\mathcal{D}})$ is minimal.}

\section{Anonymization Algorithm}\label{algorithm}

We now present COAT (COnstraint-based Anonymization of Transactions), a heuristic algorithm that solves the aforementioned problem using item generalization and suppression. Given $\mathcal{D}, \mathcal{P}, \mathcal{U}, k$ and $s$, COAT
selects a privacy constraint $p \in \mathcal{P}$, and applies
item generalizations that are specified by $\mathcal{U}$ and incur the smallest amount of information loss to satisfy $p$. When $p$ cannot be satisfied by generalization, COAT suppresses the minimum number of items in $p$ to satisfy it. The process is repeated for all privacy constraints until $\mathcal{P}$ is satisfied.

\begin{algorithm}[!ht]
\small
\begin{tabbing}
\hspace{5mm} \= \hspace{3mm} \= \hspace{3mm} \= \hspace{3mm} \= \hspace{3mm} \= \\[-3mm] \kill
1. \> $\tilde{\mathcal{D}}\gets \mathcal{D}$ \\
2. \> {\bf while} $\mathcal{P}$ is not satisfied {\bf do} \\
3. \> \> find $p$ corresponding to $I$ s.t.\\
\> \> \>$I\gets\argmin\hspace{-10mm}\displaystyle\max_{\forall p_j\in \mathcal{P} :p_j\mbox{\scriptsize is not satisfied}} \left( \mbox{sup}(\displaystyle\cup_{(i_r\in p_j)}i_r,\tilde{\mathcal{D}})\displaystyle \right)$\\
4. \> \> {\bf while} $p$ is not satisfied $~\wedge~|I|>1$ {\bf do}\\
5. \> \> \> $i_m \gets \argmin\hspace{-1mm}\displaystyle\min_{\forall i_r\in I} \mbox{sup}(i_r,\tilde{\mathcal{D}})$\\
6. \> \>  \> $u_l \gets \{u_j \in \mathcal{U}~|~i_m \in u_j\}$\\
7. \> \> \> {\bf if} $|u_l|>1$ \\
8. \> \> \> \> \emph{Generalize}$\left( i_m,u_l,\mathcal{P} \right)$\\
9. \> \> \> {\bf else if} sup$(i_m,\tilde{\mathcal{D}})<k$ \\
10. \> \> \> \> \emph{Suppress}$\left( i_m,u_l,\mathcal{P}, s\right)$\\
11.\> \> {\bf if} $p$ is not satisfied $~\wedge~|I|=1$\\
12.\> \> \>{\bf  while} $p$ is not satisfied {\bf do}\\
13. \> \> \> \> $i_m \gets \argmin\hspace{-1mm}\displaystyle\min_{\forall i_r\in I} \mbox{sup}(i_r,\tilde{\mathcal{D}})$\\
14.\> \> \> \> \emph{Suppress}$\left(i_m, u_l,\mathcal{P}, s\right)$\\
15.\> {\bf return} $\tilde{\mathcal{D}}$\\[-7mm]
\end{tabbing}
\caption{~~\small COAT($\mathcal{D},\mathcal{P},\mathcal{U}, k, s$)}\label{our_algorithm}
\end{algorithm}

Pseudocode for COAT is provided in Algorithm \ref{our_algorithm}. Since the anonymized dataset $\tilde{\mathcal{D}}$ is produced by transforming items in transactions of the original dataset $\mathcal{D}$, in step $1$, we initialize $\tilde{\mathcal{D}}$ to $\mathcal{D}$.
Steps $2$ to $14$ present the main iteration of COAT, which aims to satisfy the privacy constraint set $\mathcal{P}$ (step $2$). In step $3$, COAT selects the privacy constraint $p$ that is not satisfied and corresponds to an itemset $I$ having maximum support in $\tilde{\mathcal{D}}$, since
satisfying this constraint incurs minimal distortion of $\mathcal{D}$. This is due to the fact that the minimum number of transactions in $\tilde{\mathcal{D}}$ have to be distorted to augment the support of $I$ to at least $k$.

Next, while $p$ remains unsatisfied in $\tilde{\mathcal{D}}$ and $p$ contains more than one  generalized item (step $4$), COAT performs steps $5$ to $10$.
In step $5$, the item $i_m$ from $p$ with the minimum support in $\tilde{\mathcal{D}}$ is selected to be generalized. Selecting $i_m$ this way attempts to minimize the number of generalizations required to satisfy $p$, as items with ``low'' support
need to be generalized to meet the specified $k$. Subsequently, we identify the utility constraint $u_l$ from $\mathcal{U}$ for which $i_m \in u_l$,
in order to retrieve the items that are allowed to be generalized with $i_m$ (step $6$). If at least one item apart from $i_m$ is contained in $u_l$ (step $7$),
item $i_m$ is generalized in a way that minimizes information loss (step $8$) as illustrated in Algorithm \ref{generalize_function}. Otherwise we suppress $i_m$ through a function \emph{Suppress}, given in Algorithm \ref{suppression_function}, since applying generalization to increase the support of $i_m$ to $k$ would result in violating $u_l$ (steps $9-10$).

Steps $11$ to $14$ aim to satisfy $p$ by suppressing the minimum number of items in $\tilde{\mathcal{D}}$ required to satisfy this constraint. When $I$ consists of one (generalized) item only, and $p$ is not satisfied (step $11$), we iteratively suppress items in $I$, starting with the one having the minimum support, until $p$ is met (steps $12-14$). Last, $\tilde{\mathcal{D}}$ is released (step $15$).

\begin{algorithm}[!ht]
\small
\begin{tabbing}
\hspace{5mm} \= \hspace{3mm} \= \hspace{3mm} \= \hspace{3mm} \= \hspace{3mm} \= \\[-3mm] \kill
1. \> $i_s \gets\argmin\hspace{-6mm}\displaystyle\min_{\forall i_r\in u_l\backslash\{i_m\}}\hspace{-4.5mm}\mbox{\textit{UL}}(~(i_m,i_r)~)$\\
2. \> $\tilde{i}\gets (i_m, i_s)$\\
3.\>  {\bf foreach} $p\in \mathcal{P}~:~i_m \in p \vee i_s \in p$\\
4.\> \> $p \gets (p\cup \{\tilde{i}\}) \backslash \{i_m,i_s\}$\\
5.\>  $u_l\gets (u_l \cup \{\tilde{i}\})\backslash \{i_m,i_s\}$\\
6.\> Update transactions of $\tilde{\mathcal{D}}$ based on $\tilde{i}$\\[-7mm]
\end{tabbing}
\caption{~~\small \emph{Generalize}($i_m,u_l,\mathcal{P}$)}\label{generalize_function}
\end{algorithm}
\vspace{-5mm}
\begin{algorithm}[!ht]
\small
\begin{tabbing}
\hspace{5mm} \= \hspace{3mm} \= \hspace{3mm} \= \hspace{3mm} \= \hspace{3mm} \= \\[-3mm] \kill
1.\>  $u_l\gets u_l \backslash \{i_m\}$\\
2.\>  {\bf foreach} $p\in \mathcal{P}~:~i_m \in p$\\
3.\> \> $p \gets p \backslash \{i_m\}$\\
4.\> Remove $i_m$ from all transactions of $\tilde{\mathcal{D}}$\\
5.\> {\bf if} more than $s\%$ of items are suppressed\\
6.\> \> Error: $\mathcal{U}$ is violated\\[-7mm]
\end{tabbing}
\caption{~~\small \emph{Suppress}($i_m,u_l,\mathcal{P}, s$)}\label{suppression_function}
\end{algorithm}

Algorithms \ref{generalize_function} and \ref{suppression_function} indicate how COAT performs generalization and suppression respectively. Each of these operations involves updating the privacy and utility constraint sets $\mathcal{P}$ and $\mathcal{U}$, as well as selected transactions of $\tilde{\mathcal{D}}$.

Specifically, \emph{Generalize} (Algorithm \ref{generalize_function}) operates as follows. In step $1$, it identifies the item $i_s$ that can be generalized together with $i_m$ in a way that incurs the least possible information loss according to the $\mbox{\textit{UL}}$ measure. Step $2$ performs the mapping of the two items to a common generalized item $\tilde{i}$. Following that, steps $3-5$ update the privacy and the utility constraints to reflect this generalization. Finally, the transactions of $\tilde{\mathcal{D}}$ that supported any of $i_m, i_s$ are updated to support the generalized item $\tilde{i}$ instead.

\emph{Suppress} (Algorithm \ref{suppression_function}) involves removing an item $i_m$ from the privacy and utility constraint sets (steps $1-3$), and the transactions supporting it in $\tilde{\mathcal{D}}$ (step $4$). Finally, it checks whether the imposed suppression threshold
$s$ has been surpassed (step $5$). This happens when
utility constraints are overly restrictive (e.g., they require all items to remain intact in the anonymized dataset) and a ``low'' suppression
threshold is used. In this case, data owners are notified that the utility constraint set $\mathcal{U}$ has been violated and the anonymization process terminates (step $6$).

\begin{example}
We apply COAT on the dataset $\mathcal{D}$ of Fig. \ref{original_data}, using
the constraints of Figs. \ref{original_data_pc} and \ref{original_data_uc}, $k=5$, and $s=15\%$. Since $\mathcal{P}$ contains $p_1, p_2$ whose itemsets are equally supported in $\mathcal{D}$, COAT arbitrarily considers $p_1=\{a,b,c\}$. Then, it selects $b$, which has the minimum support among $a, b$ and $c$, and generalizes it together with $a$, as required by $u_1$. This increases the support of $(a,b)c$ to $7$, satisfying $p_1$. Subsequently, COAT considers $p_2=\{d,e,f,g,h\}$. Item $d$ is minimally supported among the items of $p_2$, thus it is considered for generalization. However, $d$ cannot be generalized due to $u_3$, and it is suppressed, since its support is below $k$. After suppressing $d$, $p_2$ is still not met. Since in $p_2$, both $g,h$ have minimum support, $g$ is arbitrarily selected to be generalized. Item $g$ can be generalized with any of $e, f$ or $h$, but it is generalized with $h$, since $(g,h)$ incurs the minimum information loss. This satisfies $p_2$ and $\mathcal{P}$ is now satisfied. $\mathcal{U}$ is also satisfied as shown in Example \ref{last_ex}.
\end{example}

\section{Specifying Privacy and Utility Constraints}\label{up_specification}

The notions of privacy and utility constraints, which
reflect itemsets deemed as potentially linkable and important for intended data analysis tasks respectively, are central to our anonymization approach.
Our constraint specification framework allows data owners to formulate detailed constraints based on their specific privacy and utility requirements, which are given as
input to COAT. However, acknowledging that constraint specification may be challenging for data owners who lack domain knowledge, we present simple methods
that aim to help such data owners formulate constraints.

Section \ref{minepc} discusses our \emph{Privacy constraint set generation} (\emph{Pgen}) algorithm that constructs a privacy constraint set automatically, assuming that attackers can use any part of any transaction to link published data to individuals. \emph{Pgen} works by searching the original dataset for itemsets with ``low'' support, each of which is treated as potentially linkable and is modeled as a privacy constraint. Although the resultant privacy constraint set corresponds to a stringent privacy policy, we believe that adopting this policy is a safe choice when data owners are unable to specify which items are potentially linkable. Section \ref{setuc1} provides a recipe to reduce the effort of specifying utility constraints.

\subsection{Constructing a Privacy Constraint Set}\label{minepc}

Before presenting \emph{Pgen}, we capture the largest part of a transaction that can be used in linking attacks using Definition \ref{max_inf_itemsets}.

\begin{definition}\label{max_inf_itemsets}\hspace{+1mm}\textsc{(Maximal Infrequent Itemsets).}
Given a transactional dataset $\mathcal{D}$, and a parameter $k$, we define the set of \emph{maximal infrequent} itemsets
in $\mathcal{D}$ as those itemsets that have a support in the interval $(0,k)$ in $\mathcal{D}$, and none of their proper supersets is
supported in $\mathcal{D}$.
\end{definition}

Example \ref{latticeexample} illustrates the above definition.

\begin{example}\label{latticeexample}
Consider a dataset comprised of the last three transactions of the dataset of Fig. \ref{original_data} (associated with the itemsets $\{a,c,f\}$, $\{a,c\}$ and $\{b,h\}$ respectively), and that  $k$ is set to $2$. The lattice of itemsets in this dataset is illustrated in Fig. \ref{examplelat}, in which the support of each supported itemset is shown next to it. As can be seen, the set of \emph{maximal infrequent} itemsets in this dataset contains only $acf$ and $bh$, as each of these itemsets is supported in the
dataset and all of its proper supersets have a support of zero.
\end{example}

\begin{figure}[!ht]
\scriptsize\centering
\includegraphics[width=0.99\columnwidth]{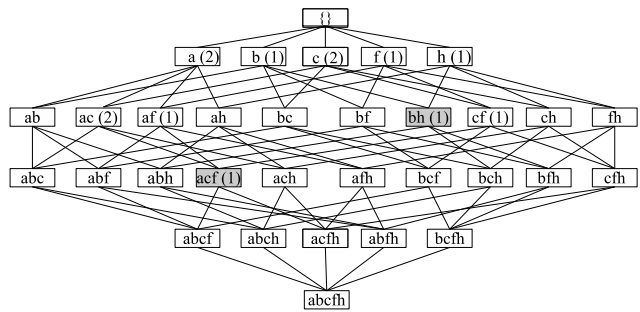}
\vspace{-3mm}
\caption{An example of \emph{maximal infrequent} itemsets}\label{examplelat}
\end{figure}

Given a transactional dataset $\mathcal{D}$, and a parameter $k$, \emph{Pgen} constructs
a privacy constraint set $\mathcal{P}$ that contains all the maximal infrequent itemsets in $\mathcal{D}$.
As mentioned above, the generated $\mathcal{P}$ can be given as input to COAT to ensure that anonymized data can prevent
linking attacks based on any part of any transaction in $\mathcal{D}$. The pseudocode of \emph{Pgen} is provided in Algorithm \ref{pgen}.

\begin{algorithm}[!ht]
\small
\begin{tabbing}
\hspace{5mm} \= \hspace{3mm} \= \hspace{3mm} \= \hspace{3mm} \= \hspace{3mm} \= \\[-3mm] \kill
1.\>  $\mathcal{P} \gets $  sorted transactions of $\mathcal{D}$ with respect to their size\\
  \> \> in decreasing order\\
2.\>  {\bf foreach} $T_r \in \mathcal{P}, r=1,...,N$\\
3.\> \> {\bf foreach} $T_s \in \mathcal{P}, s=(r+1),...,N$\\
4.\> \> \> {\bf if} $T_s \subseteq T_r$ \\
  \> \> \> \> Remove $T_s$ from $\mathcal{P}$\\
5. \> \> Itemset $I \gets T_r$\\
6.\> \> {\bf if} $sup(I,\mathcal{D})\geq k$ \\
7.\> \> \> Remove $T_r$ from $\mathcal{P}$\\
8.\> {\bf return} $\mathcal{P}$\\[-7mm]
\end{tabbing}
\caption{~~\small \emph{Pgen}($\mathcal{D},k$)}\label{pgen}
\end{algorithm}

\emph{Pgen} starts by creating a privacy constraint set $\mathcal{P}$, which is initialized by the set of transactions of the original dataset $\mathcal{D}$,
each of which is treated as a privacy constraint. Clearly, this set may contain redundant itemsets, which would result in an unnecessary computational overhead if used as input to COAT. This is because COAT works by satisfying each privacy constraint iteratively. Therefore, \emph{Pgen} implements a simple pruning strategy that removes redundant privacy constraints from $\mathcal{P}$ to reduce its size without affecting the privacy guarantees provided when $\mathcal{P}$ is satisfied.

The first step of this strategy is to populate $\mathcal{P}$ with the set of transactions of $\mathcal{D}$, sorted in terms of decreasing size. Subsequently, in steps $2$-$4$, transactions ($T_s$) that are subsets of other transactions ($T_r$) are identified and removed from $\mathcal{P}$. This is because these transactions cannot correspond to maximal infrequent itemsets, according to Definition \ref{max_inf_itemsets}.
Next, steps $6$ and $7$ ensure that privacy constraints that do not require protection (i.e. itemsets induced by transactions having a support of at least $k$ in $\mathcal{D}$) are not included in $\mathcal{P}$. Finally, $\mathcal{P}$, which contains the set of maximal infrequent itemsets in $\mathcal{D}$, is returned in step $8$. This
privacy constraint set can be given as input to COAT. Notice that \emph{Pgen} has a quadratic run-time complexity, as it involves sorting, pairwise comparison, and support computation for transactions. To illustrate how \emph{Pgen} works, we provide Example \ref{exampleseven}.

\begin{example}\label{exampleseven}
Consider applying \emph{Pgen} on the dataset of Example \ref{latticeexample}, using $k=2$. This results in
initializing $\mathcal{P}$ with three privacy constraints $p_1=\{a,c,f\}, p_2=\{a,c\}$ and $p_3=\{b,h\}$ (one for each transaction), which are sorted in terms
of decreasing size. Subsequently, $p_2$ is removed from $\mathcal{P}$, because $\{a,c\}$ is a subset of $p_1=\{a,c,f\}$. Next, \emph{Pgen} checks the support of
$p_1$, and so retains it in $\mathcal{P}$ as the support of $p_1$ in this dataset is $1\in(0,2)$. In the final iteration, \emph{Pgen} examines $p_3$, and retains it in $\mathcal{P}$ for the same reason. Thus, \emph{Pgen} returns $\mathcal{P}=\{p_1,p_3\}$.
\end{example}

\subsection{Formulating a Utility Constraint Set}\label{setuc1}

While privacy constraints can be extracted automatically as discussed above,
this is difficult for utility constraints, because they model application-specific data analysis requirements.
Thus, we assume that data owners are able to specify utility constraints to avoid distorting
itemsets that need to be used in intended applications.

When interested in generating anonymized data that allows the counts of \emph{aggregate concepts} to be accurately determined, for example,
data users can formulate a utility constraint for each of these concepts (itemsets), as explained in Section \ref{utility_constraints_section}. These itemsets may be selected with the help of hierarchies or ontologies, which are specified by domain experts or constructed in an automated fashion \cite{punera}. A utility constraint
containing the remaining items (i.e., those not contained in the selected itemsets) should also be specified to ensure that the utility constraint set is a partition of $\mathcal{I}$ (see Definition \ref{utility_constraint_set_def}).

We emphasize that the way \emph{all} items are generalized
is governed by the utility loss function (see Definition \ref{UL_item}), which forces semantically related items to be generalized together.
Example \ref{exampleeight} illustrates how utility constraints may be specified.

\begin{example}\label{exampleeight}
Consider that the dataset of Fig. \ref{original_data} has to be anonymized to support the study of Example \ref{diabetesexample} in which
the number of patients diagnosed with
\emph{diabetes} (i.e., transactions having $a$, $b$, or $ab$) needs to be accurately computed.
To support this study, the hospital can specify
a utility constraint $\{a,b\}$, and include all the remaining diagnosis codes in a second constraint $\{c,d,e,f,g,h\}$.
\end{example}

\vspace{+3mm}
\section{Experimental Evaluation}\label{experiments}

In this section, we compare COAT to Apriori \cite{140} using four series of experiments. In the first series, we compare the amount of information loss the algorithms incur to achieve $k^m$-anonymity. The second and third series of experiments examine whether the algorithms can meet detailed privacy and utility requirements without harming data utility, and the last series evaluate their efficiency.

\subsection{Experimental setup and metrics}\label{expsetupandmetr}

We use two real-world transactional datasets, \emph{BMS-WebView-1} (\emph{BMS1}) and \emph{BMS-WebView-2} (\emph{BMS2}), which contain click-stream data from two e-commerce sites. The datasets have been used in evaluating prior work \cite{140,141} and also as benchmarks in the 2000 KDD-Cup competition. Table \ref{data_desc} summarizes their characteristics.

\begin{table}[!ht]
\scriptsize\centering
\begin{tabular}{|l|l|l|l|l|}
\hline {\bf Dataset} & {\bf $N$} & $ |\mathcal{I}|$ & {\bf Max. $|T|$} & {\bf Avg. $|T|$} \tabularnewline\hline
 \emph{BMS-1} & 59602 & 497 & 267 & 2.5 \tabularnewline \hline
 \emph{BMS-2} & 77512 & 3340 & 161 & 5.0 \tabularnewline \hline
\end{tabular}
\vspace{-3mm}
\caption{Description of used datasets}\label{data_desc}
\end{table}
\vspace{-2mm}
To ensure a fair comparison between COAT and Apriori, we configured the latter with the same hierarchies as in \cite{140} and set the weights $w(\tilde{i_m})$ used in COAT based on a notion of semantic distance computed according to the aforementioned hierarchies \cite{14}. We did not compare our approach to the two other algorithms proposed by the authors of Apriori in \cite{140}. This is because these algorithms have been shown to be comparable to Apriori in terms of effectiveness, while they are only applicable to datasets with a small domain of less than $50$ items \cite{140} (typically, transactional datasets have a domain size in the order of hundreds or thousands). We also did not compare our approach to those of \cite{142} and \cite{141}, since these approaches require a fixed categorization of items into potentially linkable and sensitive, a classification that is not applicable to the problem we tackle.

Both COAT and Apriori were implemented in C++. All experiments were performed on an Intel 2.8GHz machine equipped with 4GB of RAM.

To quantify information loss, we considered aggregate query answering as an indicative application, and measured the accuracy of answering workloads of queries on anonymized data produced by the tested algorithms. This is a widely-used approach to characterize information loss \cite{8,79,141} and is invariant of the way tested algorithms work. Consider the COUNT() query $Q$ shown in Fig. \ref{query_example}. We obtain an accurate answer $a(Q)$ for $Q$ when this query is applied to original data $\mathcal{D}$, but not in the case of generalized data $\tilde{\mathcal{D}}$, as original items from $\mathcal{I}$ are mapped to generalized ones in $\tilde{\mathcal{I}}$. Therefore, we can only estimate the answer for $Q$.

\begin{figure}[!ht]
\small\centering
\begin{tabular}{ll}
Q:& \texttt{SELECT COUNT(}$T_n$\texttt{ (or }$\tilde{T_n}$\texttt{ ))}\tabularnewline
&\texttt{FROM  } $\mathcal{D}$ \texttt{(or }$\tilde{\mathcal{D}}$\texttt{)}\tabularnewline
&\texttt{WHERE }$i_1 \in T_n \wedge i_2 \in T_n \wedge \ldots \wedge i_q \in T_n$ \tabularnewline
&\texttt{~~~~~~(or~}$\Phi(i_1)\in\tilde{T_n}\wedge  ... \wedge\Phi(i_q)\in\tilde{T_n}$\texttt{)}\tabularnewline
\end{tabular}
\vspace{-3mm}
\caption{COUNT() query example}\label{query_example}
\end{figure}
\vspace{-2mm}
This estimation can be performed by computing the probability a transaction of $\tilde{\mathcal{D}}$ satisfies $Q$, as $\Pi_{r=1}^{q}p(i_r)$, where $p(i_r)$ is the probability of mapping an item $i_r$, $r=1,...,q$, in the query to a generalized item $\tilde{i_m}$, assuming that $\tilde{i_m}$ can include any possible subset of the items mapped to it with equal probability, and that there are no correlations among generalized items \cite{8,79,141}. An estimated answer $e(Q)$ of $Q$ is then derived by summing the corresponding probabilities across all transactions $\tilde{T_n}$ of $\tilde{\mathcal{D}}$.

To measure the accuracy of estimating $Q$, we use the \emph{Relative Error} (\emph{RE}) measure computed as \emph{RE}($Q$)$=|a(Q)-e(Q)|/a(Q)$. Given a workload of queries, the Average Relative Error (\emph{AvgRE}) for all queries, reflects how well anonymized data supports query answering \cite{8,79}. To measure \emph{AvgRE}, we constructed workloads comprised of $1000$ COUNT() queries similar to $Q$. The items participating in these queries were selected randomly from the generalized items.

\subsection{Achieving $\mathbf{k^m}$-anonymity}\label{achieving_km}

In this section, we empirically confirm that COAT not only satisfies $k^m$-anonymity, but does so with up to $9$ times less information loss than Apriori. Specifically, we ran COAT by including all $m$-itemsets in the privacy constraint set $\mathcal{P}$, considering a utility
constraint set $\mathcal{U}$ that contains all items (effectively allowing all possible generalizations), and setting $s=0.5\%$. Both algorithms used the same $k$ and $m$ values. The  results with respect to \emph{AvgRE} and \textit{UL} measures are summarized in Sections \ref{achieving_km_re} and \ref{achieving_km_um} respectively. In these experiments COAT did not suppress any items.

\subsubsection{Capturing data utility using AvgRE}\label{achieving_km_re}

Figs. \ref{H_Q2_M1} and \ref{J_Q3} report \emph{AvgRE} scores for \emph{BMS1}, where the number of items $q$ included in a query was $1$ and $3$
respectively, $m$ was set to $2$, and $k$ was selected over the range $[2,50]$. As expected, increasing $k$ induced more information loss due to the utility/privacy trade-off. Increasing $q$ had a similar effect because accurately answering queries involving many items is more difficult. COAT outperformed Apriori in both cases, achieving up to $9$ times better \emph{AvgRE} scores. This is because, as $k$ increases, the recoding model of Apriori forces an increasingly large number of items to be generalized together, while the model in COAT generalizes no more items than required to protect an itemset. Similar results were achieved for \emph{BMS2} (omitted for brevity).

\begin{figure}[!ht]
\subfigcapskip=-5pt
\subfigure[]
{
\includegraphics[width=0.47\columnwidth]{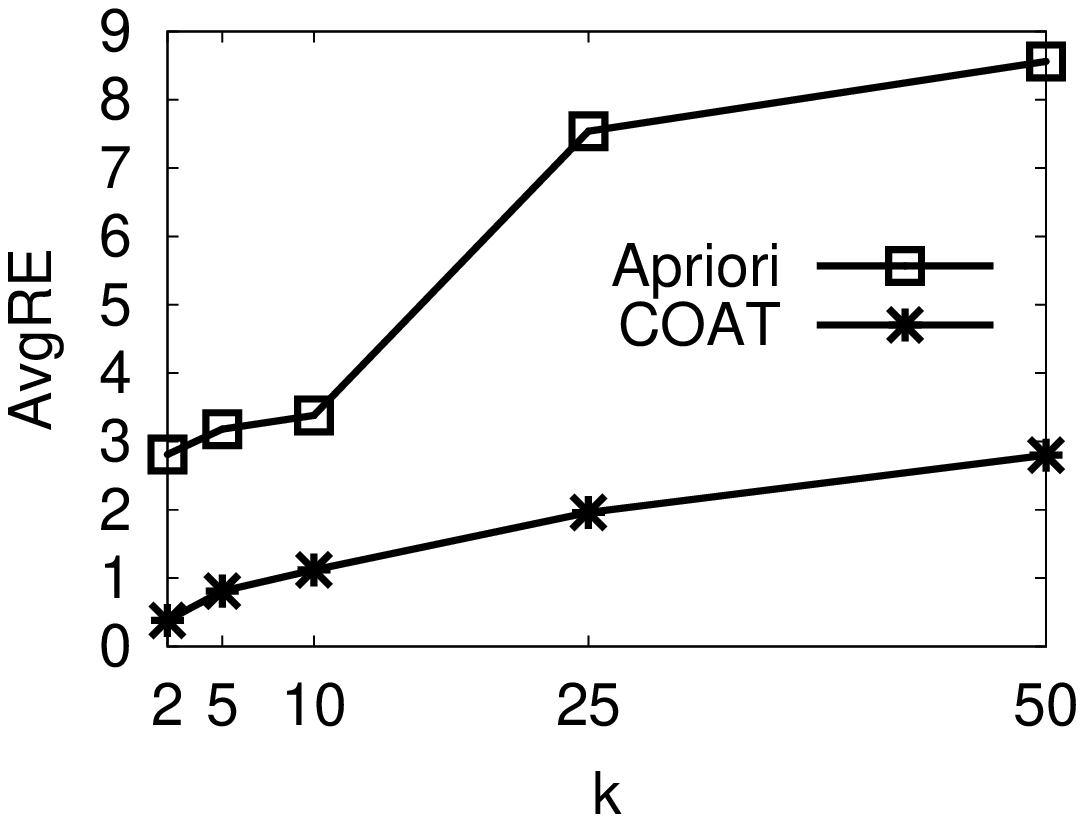}\label{H_Q2_M1}
}
\subfigcapskip=-5pt
\subfigure[]
{
\includegraphics[width=0.47\columnwidth]{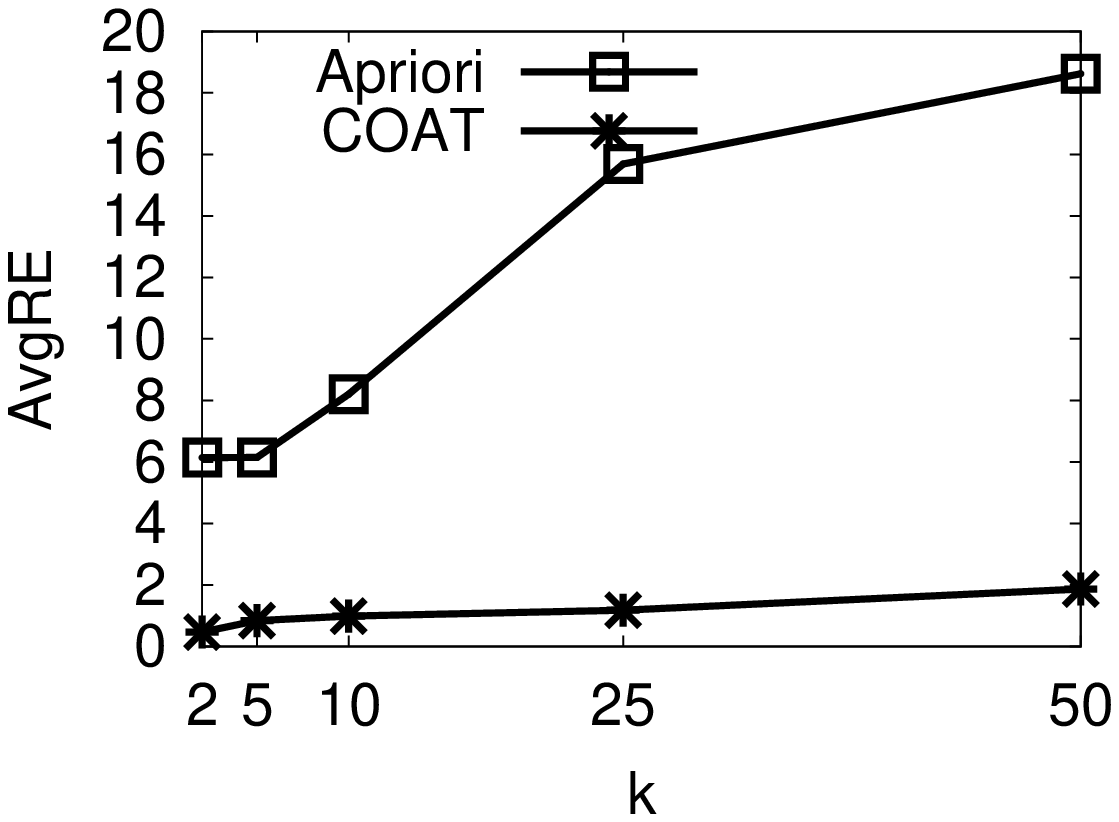}\label{J_Q3}
}
\vspace{-3mm}
\caption{AvgRE vs. $k$ for (a) $q=1$ and (b) $q=3$ in \emph{BMS1}}\label{exp1}
\end{figure}

We also executed COAT and Apriori using $k=5$, and varied $m$ between $1$ and $3$. The \emph{AvgRE} scores for \emph{BMS1} are shown in Fig. \ref{L_m123}. Apriori incurred $7$ times more information loss than COAT to anonymize \emph{BMS1} when $m=3$. This is because the number of items that Apriori forces to be generalized together to protect $m$-itemsets increases substantially as $m$ grows. The impact of this generalization strategy on data utility was even more evident in the case of \emph{BMS2}, as shown in Fig. \ref{G_K5_Q2}.

\begin{figure}[!ht]
\subfigcapskip=-5pt
\subfigure[]{
\includegraphics[width=0.47\columnwidth]{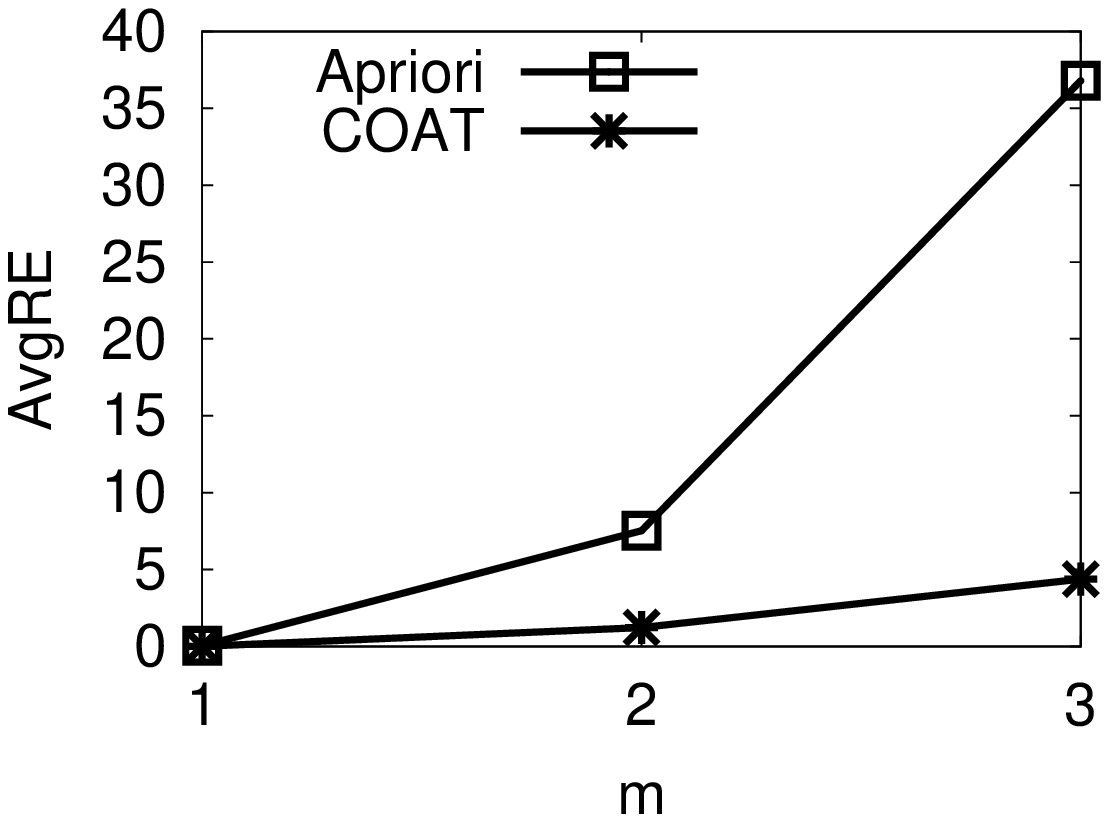}\vspace{-4mm}
\label{L_m123}
}
\subfigcapskip=-5pt
\subfigure[]{
\includegraphics[width=0.47\columnwidth]{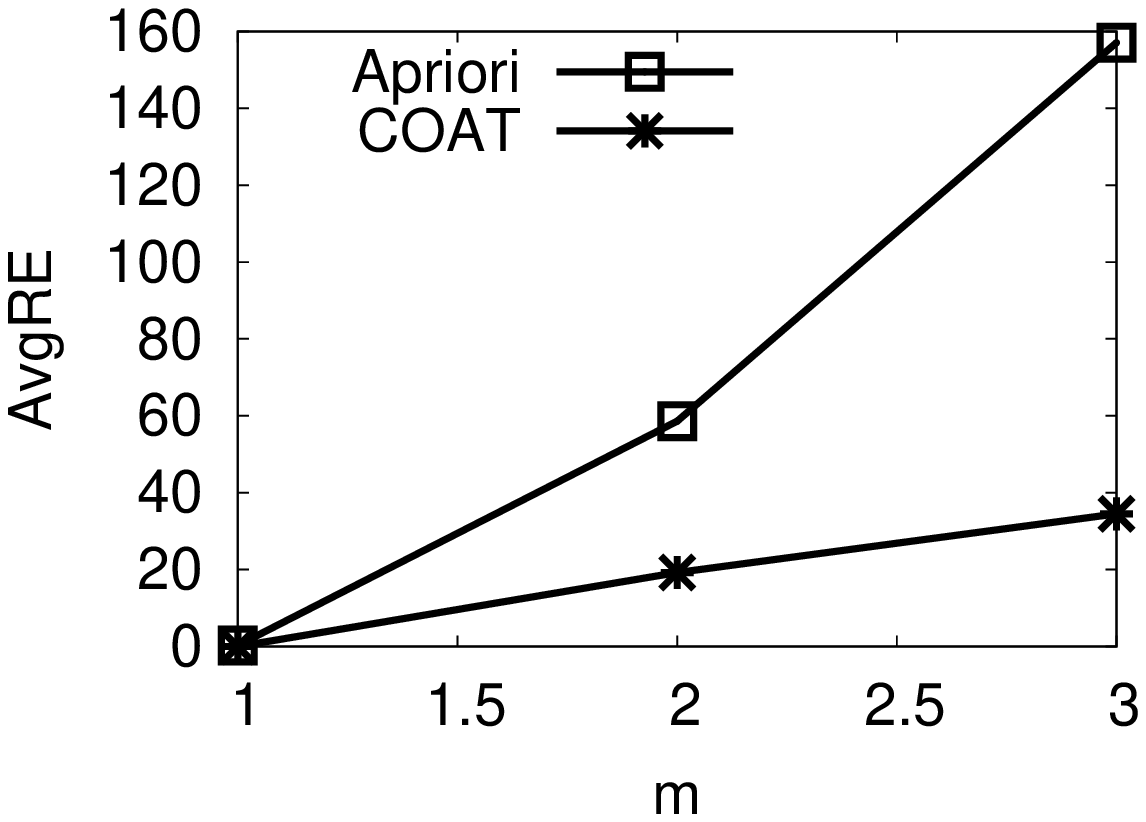}\vspace{-4mm}
\label{G_K5_Q2}
}
\vspace{-3mm}
\caption{AvgRE vs. $m$ for (a) \emph{BMS1} and (b) \emph{BMS2}}\label{exp2}
\end{figure}

\subsubsection{Capturing data utility using UL}\label{achieving_km_um}

We compared the two algorithms with respect to the \textit{UL} measure. Fig. \ref{PARTA1} shows the result of running these algorithms on \emph{BMS2} using $m=2$ and $k$ values between $2$ and $50$. Observe that Apriori was fairly insensitive to $k$ up to $25$. In fact, Apriori over-generalized itemsets by increasing their support to much larger values than $k$ due to its recoding strategy. On the other hand, COAT achieved a much better result for all tested $k$ values, due to the fine-grained generalization model it employs. We also examined how the algorithms fared with respect to \textit{UL} when $m$ varies between $1$ and $3$, and $k=5$. Observe that Apriori incurred substantially more information loss than COAT for all tested $m$ values. This again suggests that the generalization scheme of Apriori distorts data much more than our set-based anonymization strategy. Similar results were obtained for \emph{BMS1} (omitted for brevity).

We do not report additional results with respect to \textit{UL} because COAT is designed to optimize this measure, and thus outperformed Apriori in all tested cases.

\begin{figure}[!ht]
\subfigcapskip=-5pt
\subfigure[]{
\includegraphics[width=0.47\columnwidth]{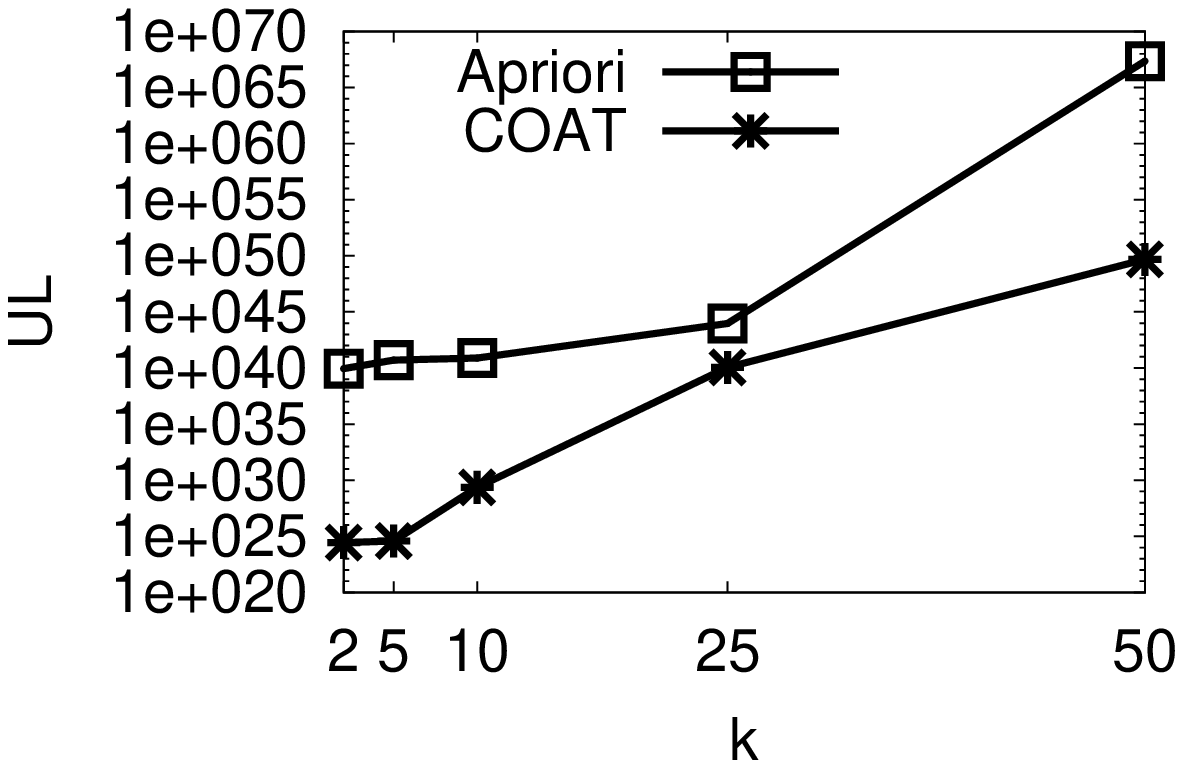}\label{PARTA1}
}
\subfigcapskip=-5pt
\subfigure[]{
\includegraphics[width=0.47\columnwidth]{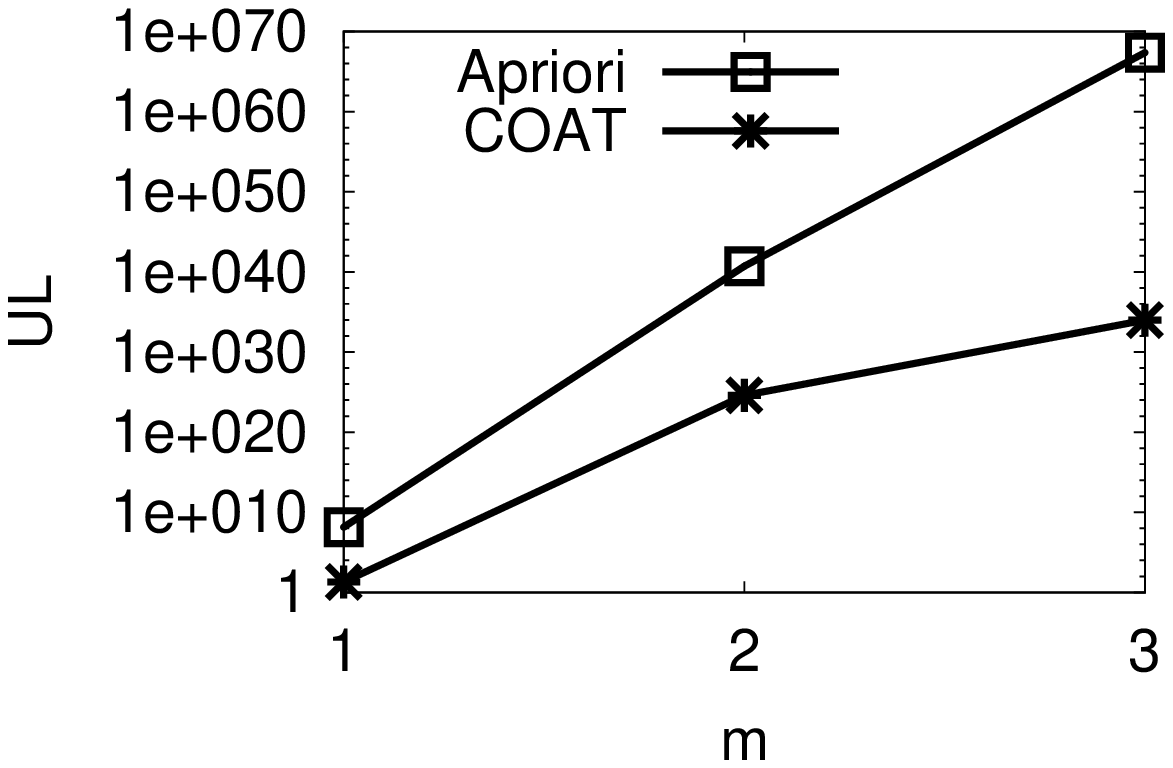}\label{PARTA2}
}
\vspace{-3mm}
\caption{(a) UL vs. $k$ and (b) UL vs. $m$ for \emph{BMS2}}\label{achieving_km_um_k}
\end{figure}

\subsection{Privacy constraints vs. data utility}\label{impact_of_pc}

In this section, we experimentally confirm that our approach can generate anonymizations with a high level of data utility through the specification of detailed privacy constraints by data owners. The impact of constraints generated by \emph{Pgen} on data utility will be examined in Section \ref{casestudy}. We constructed two types of privacy policies to simulate different privacy requirements, one in which itemsets that require protection are all of the same size and comprised of certain items from $\mathcal{I}$, and another in which such itemsets differ in size. The utility constraint set $\mathcal{U}$ used in COAT was set as in Section \ref{achieving_km}.

\subsubsection{Protecting itemsets comprised of certain items}

We considered $5$ privacy policies of the first type, \emph{PP1},\dots, \emph{PP5}, each of which assumes that all $2$-itemsets containing a certain percent of
randomly selected items require protection with $k=5$. The mappings between privacy policies and the percent of such items are
as follows: \emph{PP1}$\rightarrow 2\%$, \emph{PP2}$\rightarrow 5\%$, \emph{PP3}$\rightarrow 10\%$, \emph{PP4}$\rightarrow 25\%$, \emph{PP5}$\rightarrow 50\%$. These policies are taken into account by COAT, but not by Apriori, which needs to protect all $2$-itemsets to satisfy them.

We first studied how privacy policies affect data utility, as captured
by \emph{AvgRE}. Figs. \ref{ccc} and \ref{ddd} illustrate the results for $q=1$ and $q=3$ respectively. As expected, because it avoids unnecessarily protecting itemsets that are not specified by these policies, COAT distorted data significantly less than Apriori. This is supported by the \emph{AvgRE} scores for COAT which were significantly better than Apriori. Furthermore, as policies become more strict (i.e., require protecting itemsets induced by a larger percent of items from $\mathcal{I}$), the \emph{AvgRE} scores for COAT became slightly worse due to the utility/privacy traded-off. Nevertheless, these scores remain substantially better than that of Apriori in all cases. We repeated the same experiments for \emph{BMS2}, and obtained similar results shown in Figs. \ref{out2xx} and \ref{out6} respectively.

\begin{figure}[!ht]
\subfigcapskip=-5pt
\subfigure[]{
\includegraphics[width=0.47\columnwidth]{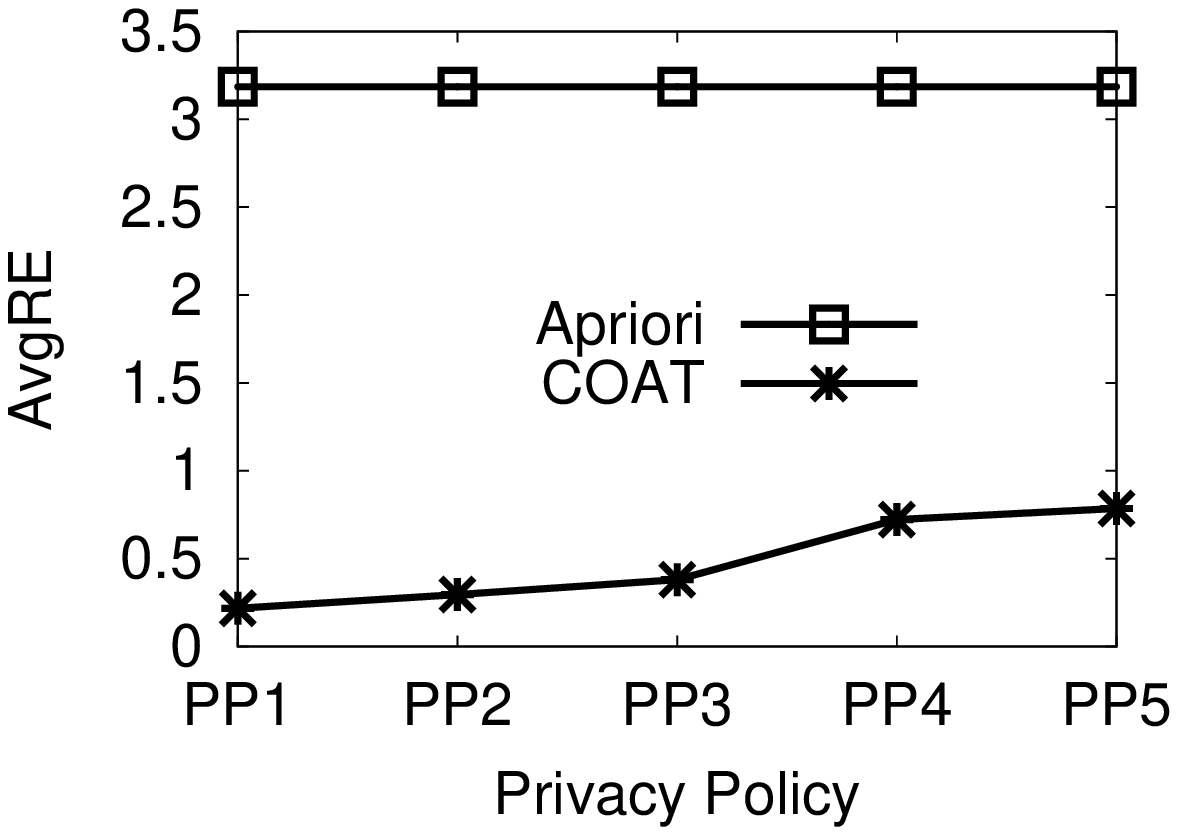}\label{ccc}
}
\subfigcapskip=-5pt
\subfigure[]{
\includegraphics[width=0.47\columnwidth]{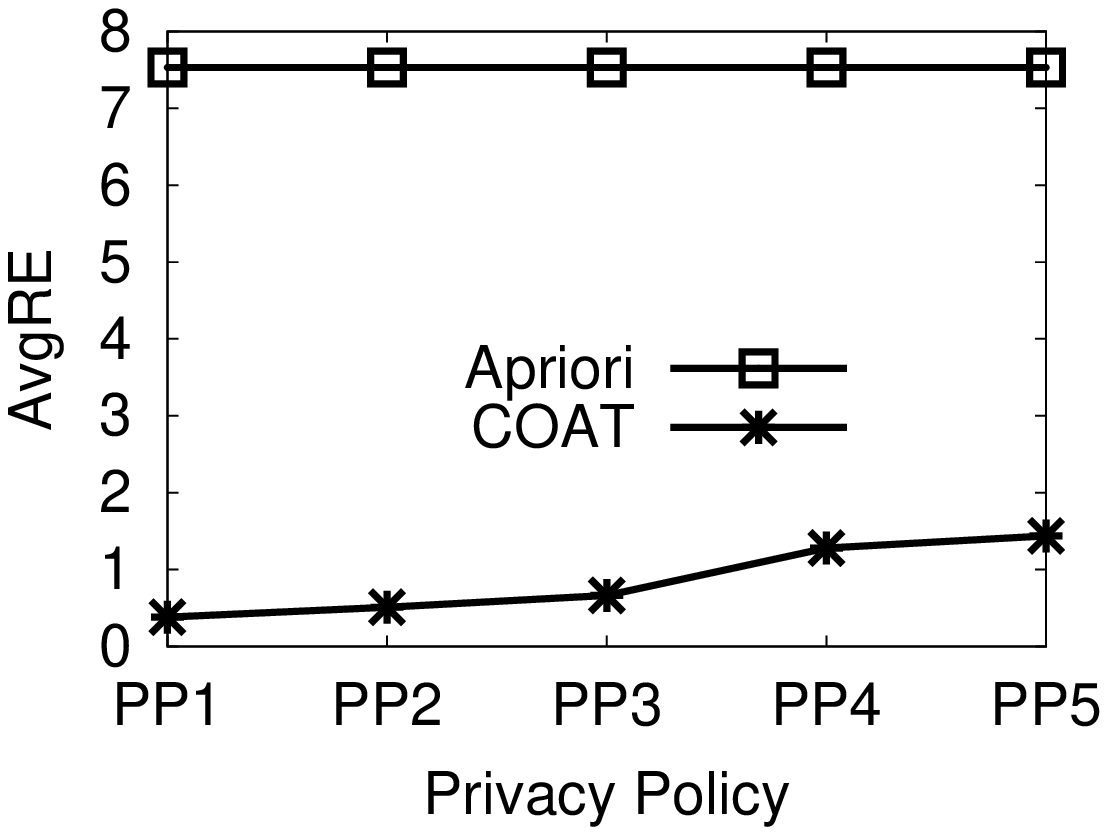}\label{ddd}
}
\subfigcapskip=-5pt
\subfigure[]{
\includegraphics[width=0.47\columnwidth]{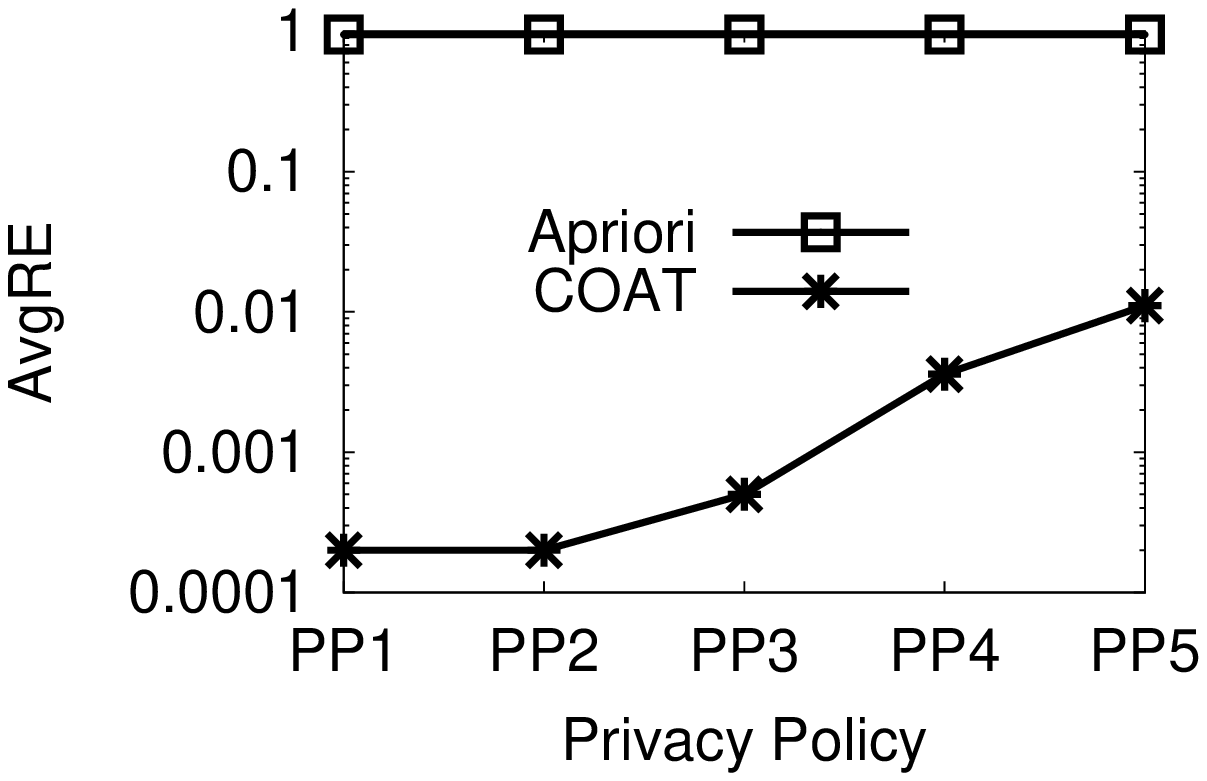}\label{out2xx}
}
\subfigcapskip=-5pt
\subfigure[]{
\includegraphics[width=0.47\columnwidth]{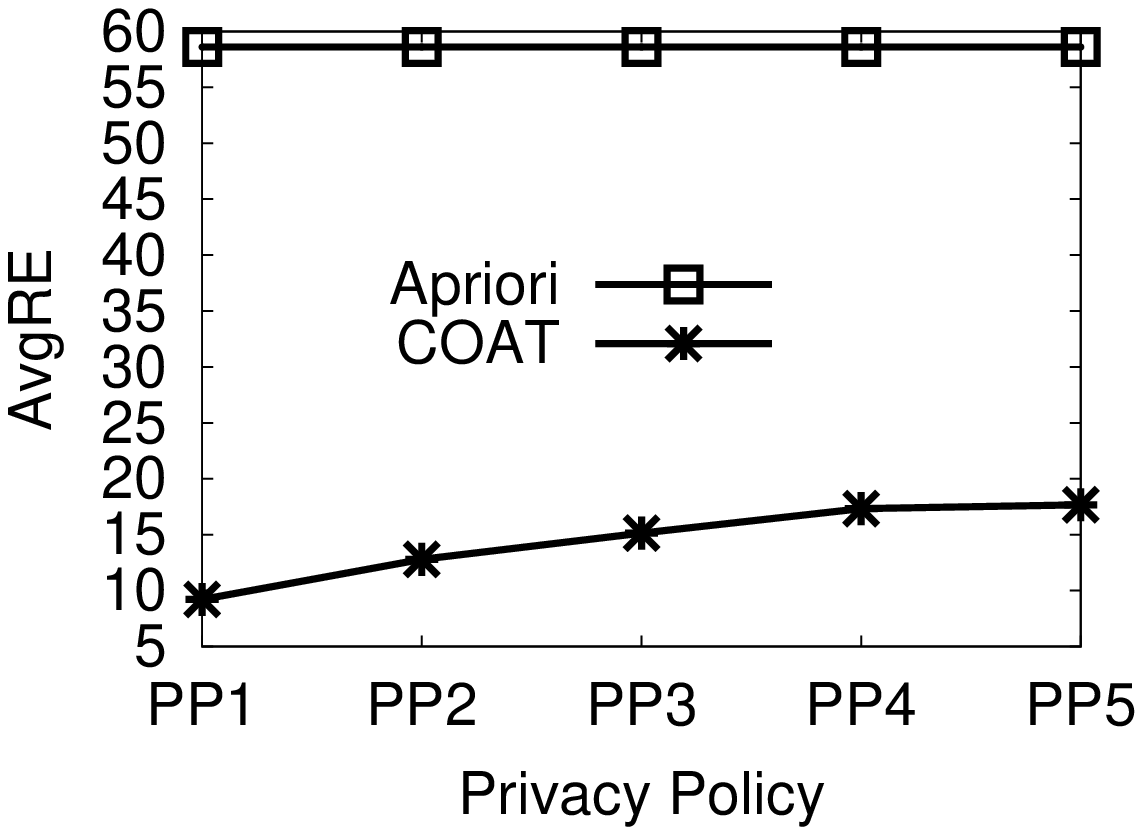}\label{out6}
}
\vspace{-3mm}
\caption{AvgRE vs. Privacy Policy (a) for $q=1$ and (b) for $q=3$ (in \emph{BMS1}), and (c) for $q=1$ and (d) for $q=3$ (in \emph{BMS2})}\label{impact_of_pc_re_bms2} \label{impact_of_pc_re_bms1}
\end{figure}

\subsubsection{Protecting itemsets of varying size}

We simulated $4$ privacy policies of the second type: \emph{PP6}, \dots, \emph{PP9}. In each of these policies, $\mathcal{P}$ consisted of itemsets with size $1$ to $4$, as shown in Table \ref{itemsets_varlen_perc}, and $k=5$. To account for these policies, Apriori had to protect all possible $4$-itemsets, and thus it was configured with $m=4$.

\begin{table}[!ht]
\scriptsize

\centering
\begin{tabular}{|l||l|l|l|l|}
\hline Privacy & $\%$ of   & $\%$ of   & $\%$ of   & $\%$ of \tabularnewline
    Policy  & items & $2$-itemsets & $3$-itemsets & $4$-itemsets \tabularnewline\hline\hline
 \emph{PP6} & $33\%$ & $33\%$ & $33\%$ & $1\%$ \tabularnewline \hline
 \emph{PP7} & $30\%$ & $30\%$ & $30\%$ & $10\%$ \tabularnewline \hline
 \emph{PP8} & $25\%$ & $25\%$ & $25\%$ & $25\%$ \tabularnewline \hline
 \emph{PP9} & $16.7\%$ & $16.7\%$ & $16.7\%$ & $50\%$ \tabularnewline \hline
\end{tabular}
\vspace{-3mm}
\caption{Summary of privacy policies \emph{PP6}, \dots , \emph{PP9}}\label{itemsets_varlen_perc}
\end{table}

The \emph{AvgRE} scores for \emph{BMS1} and \emph{BMS2}, and a workload comprised of queries with $q=2$ are depicted in Figs. \ref{XY1} and \ref{XY2} respectively.
Notice that COAT achieved better \emph{AvgRE} scores in both datasets, permitting answers to queries up to $40$ times more accurately than Apriori. This is because COAT applies generalization to each privacy constraint separately, thereby applying the minimum level of generalization required to satisfy the specified constraint.

\begin{figure}[!ht]
\subfigcapskip=-5pt
\subfigure[]{
\includegraphics[width=0.47\columnwidth]{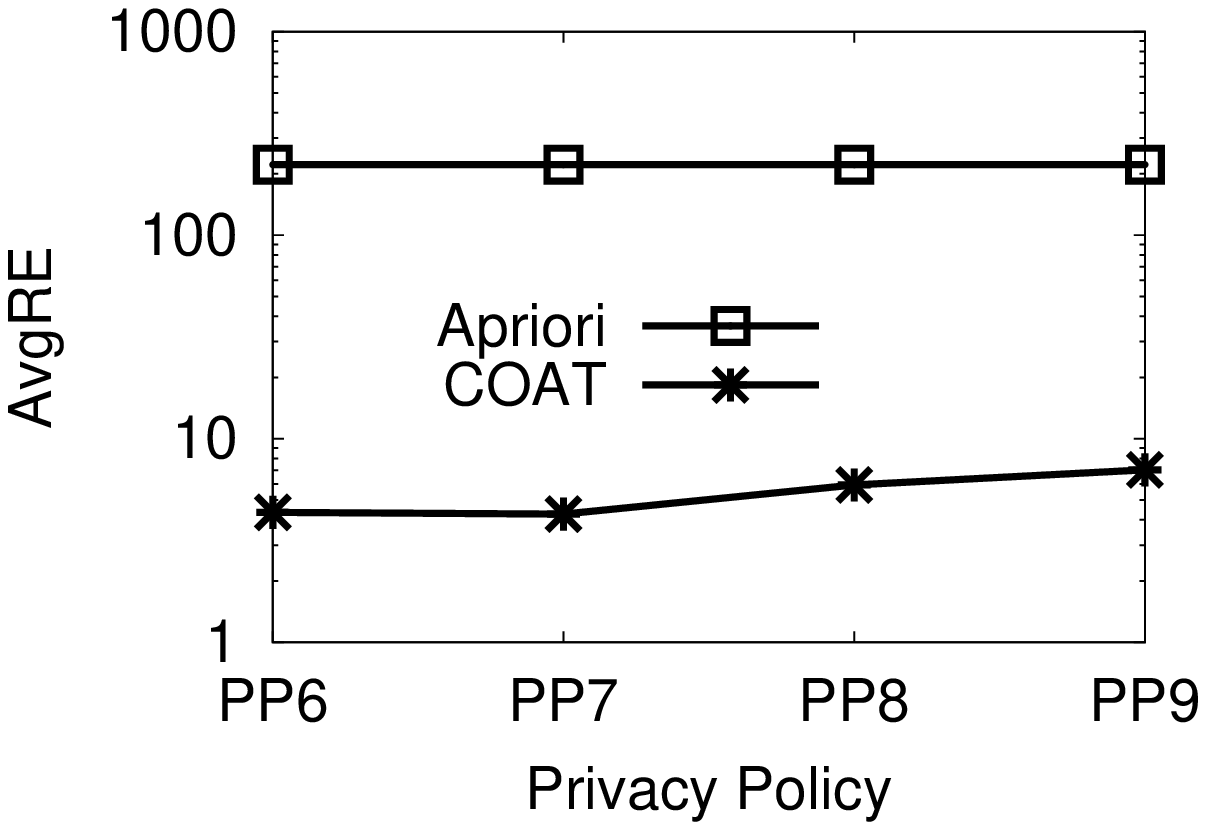}\label{XY1}
}
\subfigcapskip=-5pt
\subfigure[]{
\includegraphics[width=0.47\columnwidth]{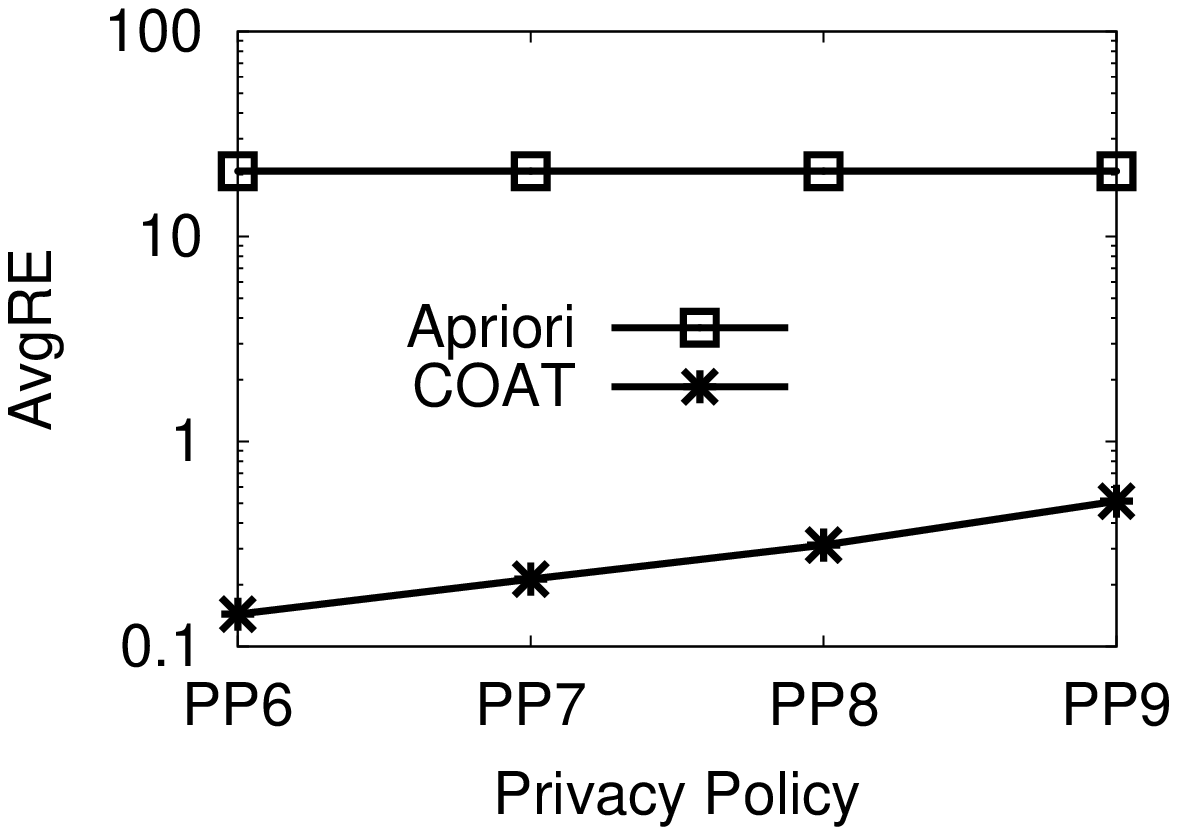}\label{XY2}
}
\vspace{-3mm}
\caption{AvgRE vs. Privacy Policy for $q=2$ (a) in \emph{BMS1} and (b) in \emph{BMS2}}\label{XYs}
\end{figure}

\subsection{Utility constraints vs. data utility}\label{impact_of_uc}

The experiments reported in this section examine the effect of utility constraints on data utility. We assumed $4$ utility policies: \emph{UP1}, \dots , \emph{UP4}. Each policy contains groups of a certain number of semantically close items (i.e., sibling items in the hierarchy). The mappings between utility policies and the size of
these groups are as follows: \emph{UP1}$\rightarrow 25$, \emph{UP2}$\rightarrow 50$, \emph{UP3}$\rightarrow 250$, and \emph{UP4}$\rightarrow 500$. Items in each group are allowed to be generalized together. Note that \emph{UP1} and \emph{UP2}, which have smaller group sizes, are very stringent and may require suppression to be satisfied. For this reason, we configured COAT with a small suppression threshold $s$ of $0.5\%$. Apriori does not address these policies because item generalization is not guided by utility constraints. Also, the privacy constraint set $\mathcal{P}$ included all $2$-itemsets and Apriori was run with $m=2$.

\emph{AvgRE} scores for a workload of queries with $q=1$ and $q=3$, are shown in Figs. \ref{UPC} and \ref{UPD} respectively, for \emph{BMS1}. Observe that COAT significantly outperformed Apriori for all utility policies. Furthermore, the number of suppressed items was very small ($0.01\%$) and occurred only in the case of \emph{UP1}. This illustrates the effectiveness of COAT, which suppresses the minimum number of items required, and only when utility constraints cannot be otherwise met. We also note that COAT was able to satisfy the imposed utility policies in all cases, unlike Apriori which was unable to meet any of them. Interestingly, the \emph{AvgRE} scores for COAT were not substantially affected by utility policies. This is because COAT applied a much lower level of generalization than that specified by the utility constraints. Similar trends were observed for \emph{BMS2} (omitted for brevity).

\begin{figure}[!ht]
\subfigcapskip=-5pt
\subfigure[]{
\includegraphics[width=0.47\columnwidth]{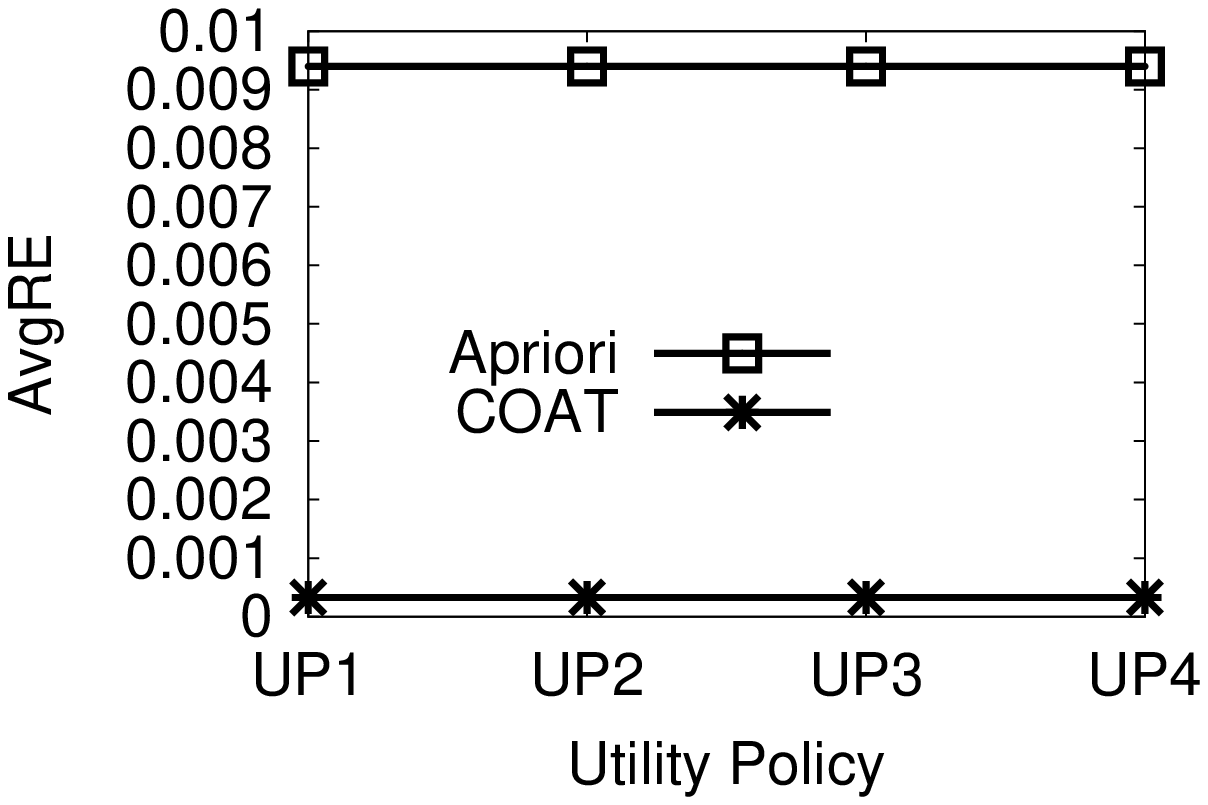}\label{UPC}
}
\subfigcapskip=-5pt
\subfigure[]{
\includegraphics[width=0.47\columnwidth]{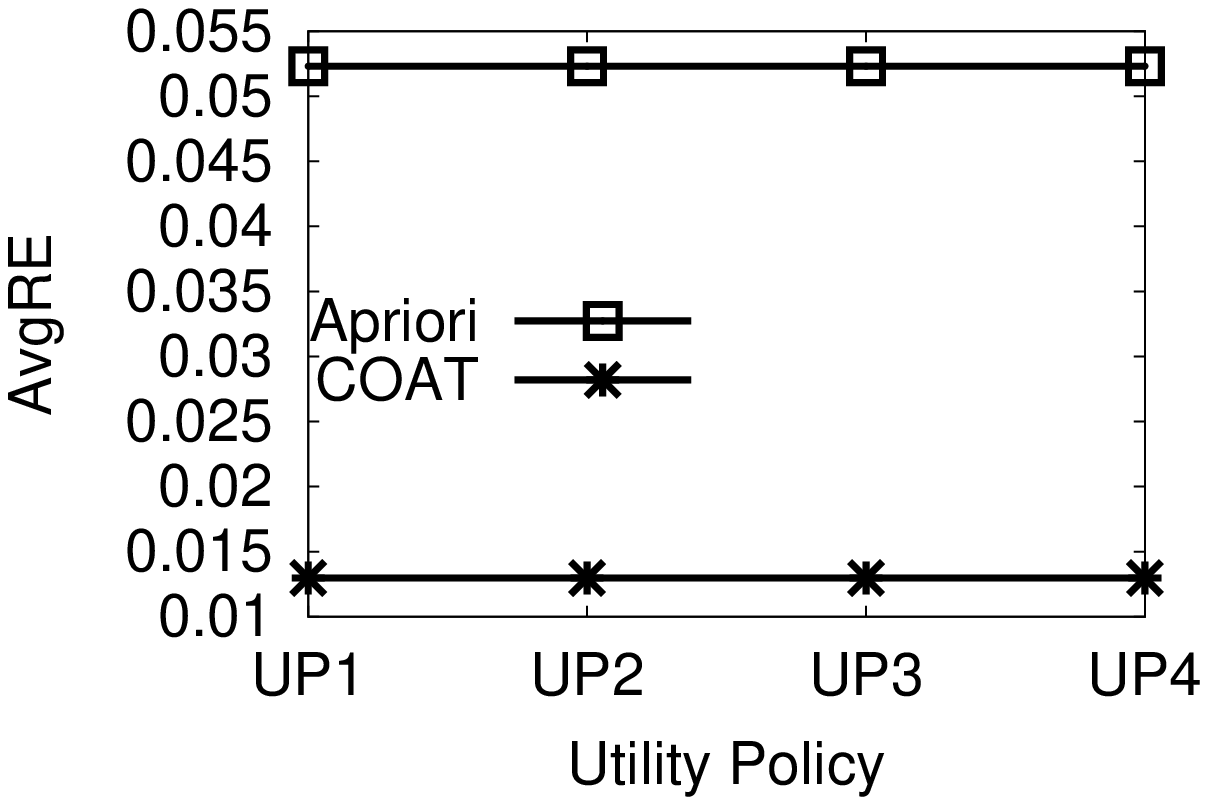}\label{UPD}
}
\vspace{-3mm}
\caption{AvgRE vs. Utility Policy (a) for $q=1$ and (b) for $q=3$ (in \emph{BMS1})}
\end{figure}


\subsection{Efficiency of Computation}

We compared COAT and Apriori in terms of efficiency.
We first examined the scalability of these algorithms with respect to dataset cardinality, by applying them on a dataset constructed by randomly selecting transactions of
\emph{BMS1}. COAT was configured by setting $\mathcal{P} $ and $\mathcal{U}$ as in Section \ref{achieving_km}, $m=2$ and $k=5$. Apriori was run with the same $k$ and $m$ values.
Fig. \ref{efD} reports run-time as cardinality varies from 1K to 50K transactions. COAT scales better than Apriori with the size of the dataset;
up to $2.5$ times faster. This is because COAT prunes the space by discarding protected itemsets as cardinality increases, whereas Apriori considers all $m$-itemsets as well their possible generalizations.

\begin{figure}[!ht]
\subfigcapskip=-5pt
\subfigure[]
{\includegraphics[width=0.48\columnwidth]{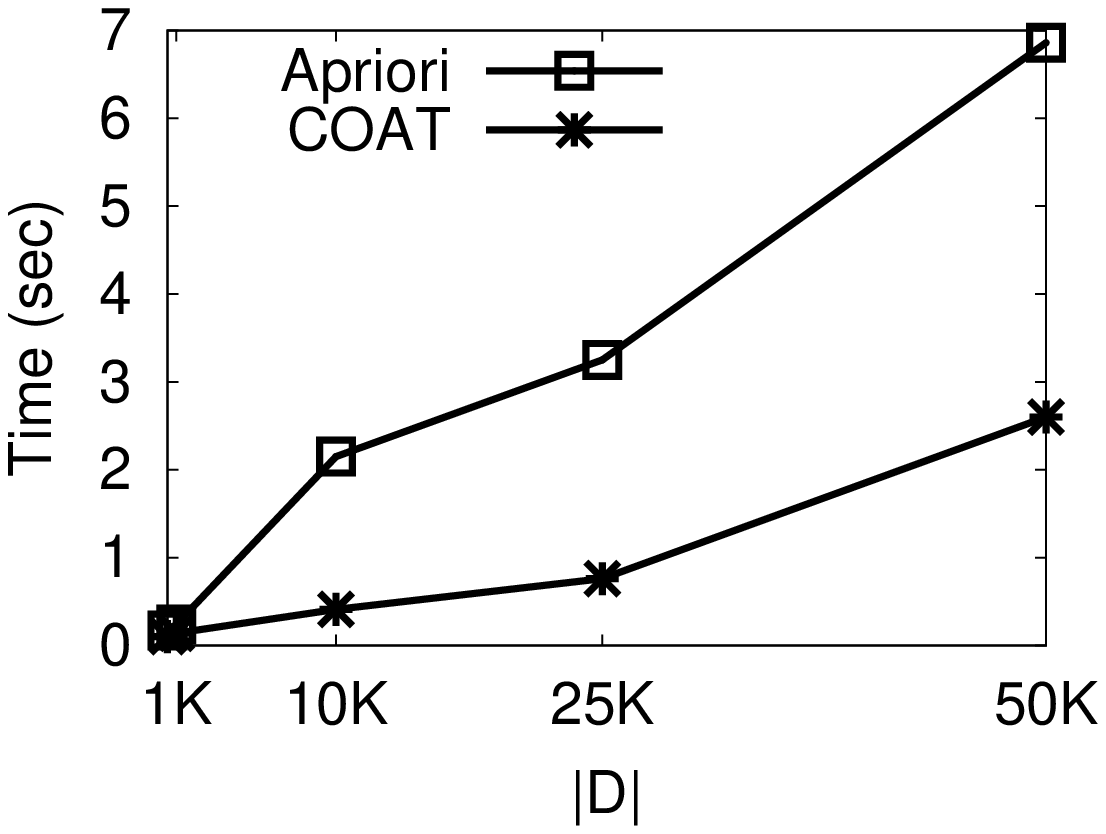}\label{efD}
}
\subfigcapskip=-5pt
\subfigure[]
{
\includegraphics[width=0.48\columnwidth]{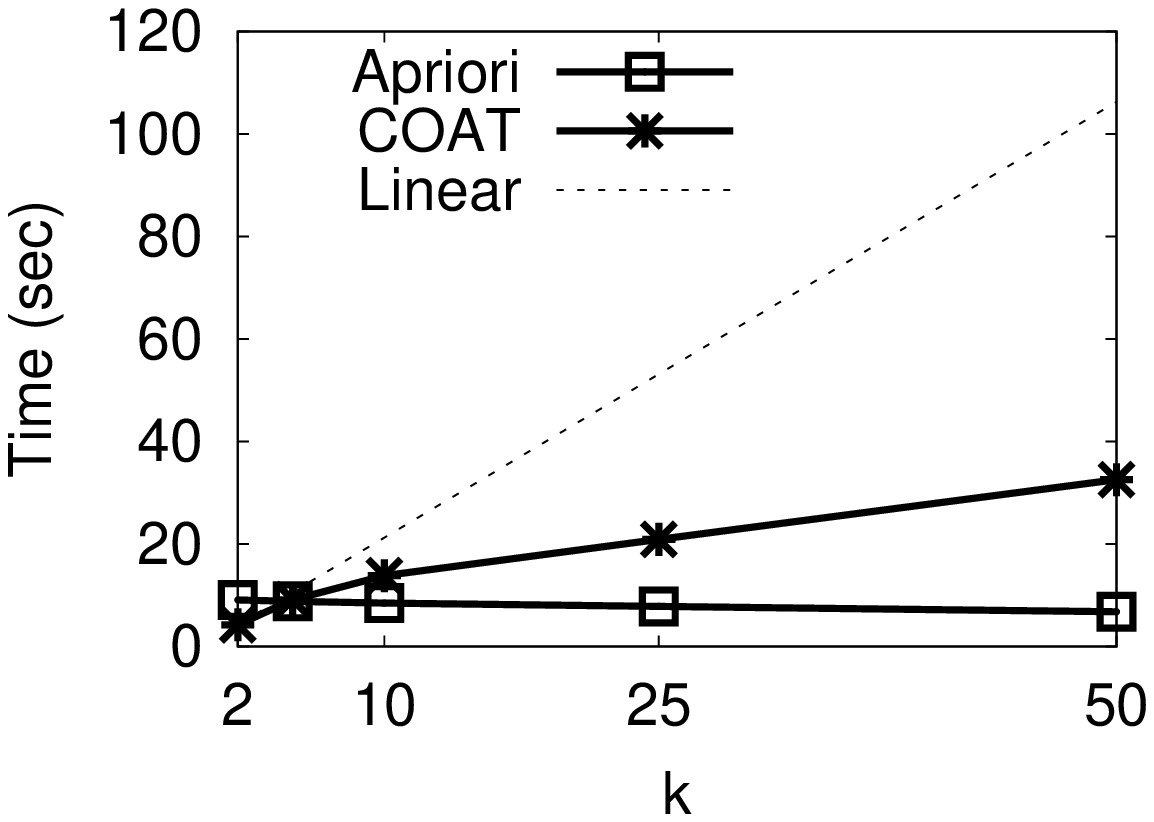}\label{efM}
}
\vspace{-3mm}
\caption{Efficiency vs. (a) dataset size $|D|$ and (b) $k$}
\end{figure}

Last, we evaluated the impact of $k$ on the run-time of COAT and Apriori on \emph{BMS1}. We used $k$ values between $2$ and $50$, and set up all other parameters as in the previous experiment. As can be seen in Fig. \ref{efM}, COAT is slightly less efficient than Apriori. This is due to the fact that COAT generalizes one item at a time, exploiting the flexibility of the set-based anonymization model. By comparison, Apriori generalizes entire subtrees of items and thus reaches the specified $k$ faster. Nevertheless, the computation cost of COAT was less than half a minute, remaining sub-linear for all testes values of $k$.

\section{Case Study: Diagnosis Codes}\label{casestudy}
In this section, we examine whether COAT can produce anonymized data that permits accurate analysis in a real-world scenario
involving detailed, application-specific utility requirements. In this context, a transactional dataset (referred to as \emph{EMR}) derived from the Electronic Medical Record system of the Vanderbilt University Medical Center \cite{Stead} needs to be published to enable certain biomedical studies. Each transaction of \emph{EMR} corresponds to a distinct patient, and contains his/her diagnosis codes in the form of \emph{ICD-9} codes \footnote{ICD-9 is the official system of assigning codes to diagnoses in the U.S.}. Table \ref{case_study_tab} summarizes the characteristics of \emph{EMR}.

\begin{table}[!ht]
\scriptsize\centering
\begin{tabular}{|l|l|l|l|l|}
\hline {\bf Dataset} & {\bf $N$} & $ |\mathcal{I}|$ & {\bf Max. $|T|$} & {\bf Avg. $|T|$} \tabularnewline\hline
 \emph{EMR} & 1336 & 5830 & 25 & 3.1 \tabularnewline \hline
\end{tabular}
\vspace{-3mm}
\caption{Description of the \emph{EMR} dataset.}\label{case_study_tab}
\end{table}

The studies that anonymized data needs to support focus on $20$ different disorders, each of which is modeled as a set of ICD-9 codes. For instance, \emph{pancreatic cancer} is represented as a set of $7$ ICD-9 codes, which correspond to different forms of \emph{pancreatic cancer} and indicate that a patient suffers from this disorder. To support these studies, the number of patients suffering from each of these disorders needs to be accurately computed. At the same time, the linkage of transactions to patients' identities based on any combination of ICD-9 codes must be prevented, because the vast majority of ICD-9 codes contained in \emph{EMR} can be found in other sources, as verified in our previous study \cite{LoukidesAMIA}.

To achieve both privacy and utility, we used our \emph{Pgen} algorithm to construct a privacy constraint set, and formulated a utility constraint set comprised of $20$ utility constraints, each for a different disorder (e.g., we specified a utility constraint that contains the $7$ ICD-9 codes corresponding to \emph{pancreatic cancer}). Furthermore, we configured COAT by setting the weights $w(\tilde{i_m})$ used in it based on a notion of semantic similarity \cite{14} computed according to the hierarchy for ICD-9 codes \footnote{http://www.cdc.gov/nchs/icd/icd9cm.htm}, and limited the maximum allowable fraction of suppressed items by setting $s$ to $0.5\%$. Apriori was also applied to anonymize \emph{EMR}, although it provides no guarantees that utility constraints are satisfied.

We evaluated the utility of anonymizations produced by both COAT and Apriori in two ways. First, we examined whether the produced anonymizations satisfied the specified utility constraint set. In fact, anonymizations constructed by COAT satisfy the latter set for all tested $k$ values (namely $2,5,10,25$ and $50$).
Thus, COAT managed to generate practically useful anonymizations that allow the number of patients having any of the $20$ disorders used in the intended studies
to be accurately computed (see Corollary \ref{uc_coro}). On the other hand, the anonymizations constructed by the Apriori algorithm did not satisfy the specified utility constraint set for any of the tested $k$ values. Therefore, we did not evaluate the data utility of anonymizations produced by Apriori using other criteria.

In addition to satisfying the specified utility constraints, it is also important to generate anonymized data with ``low'' information loss
that can support general data analysis tasks. Therefore, we investigated whether our method can generate anonymizations that are useful in aggregate query answering. To capture the amount of information loss,
we used the \emph{AvgRE} measure, discussed in Section \ref{expsetupandmetr}. \emph{AvgRE} was computed using two
different workloads referred to as \emph{W1} and \emph{W2} respectively. \emph{W1} is comprised of COUNT() queries that retrieve combinations of ICD-9 codes supported by at least $10\%$ of the transactions of \emph{EMR}. These combinations correspond to frequently co-occurring disorders (e.g., \emph{diabetes}
and \emph{hypertension}) that are important in the context of biomedical data analysis, and are different from
the $20$ disorders contained in the utility constraint set. \emph{W2} is similar to the workload considered in Section \ref{expsetupandmetr}. It is comprised of $1000$ COUNT() queries similar to the query shown in Fig. \ref{query_example}, each of which is comprised of $2$ ICD-9 codes
randomly selected among generalized items. This workload models a scenario involving anonymized data queried by users with various data analysis requirements.

\begin{figure}[!ht]
\subfigcapskip=-5pt
\subfigure[]
{
\includegraphics[width=0.47\columnwidth]{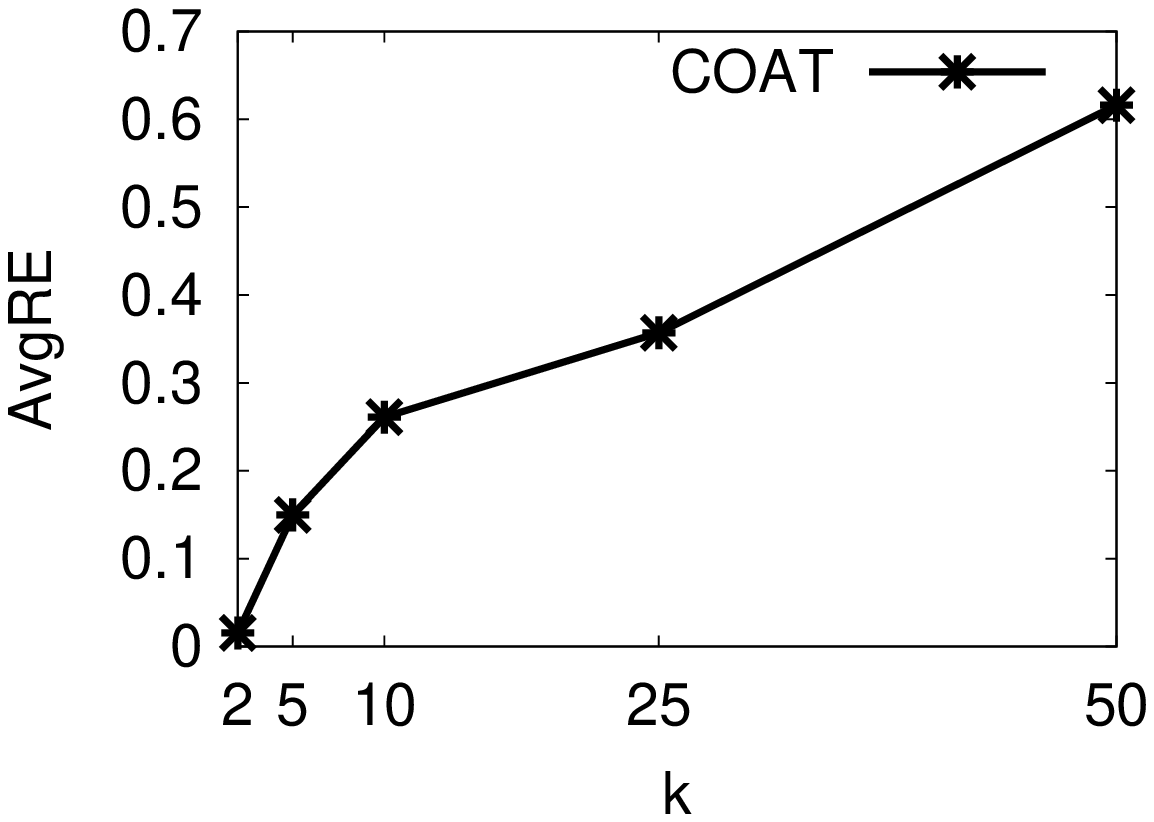}\label{case1}
}
\subfigcapskip=-5pt
\subfigure[]
{\includegraphics[width=0.47\columnwidth]{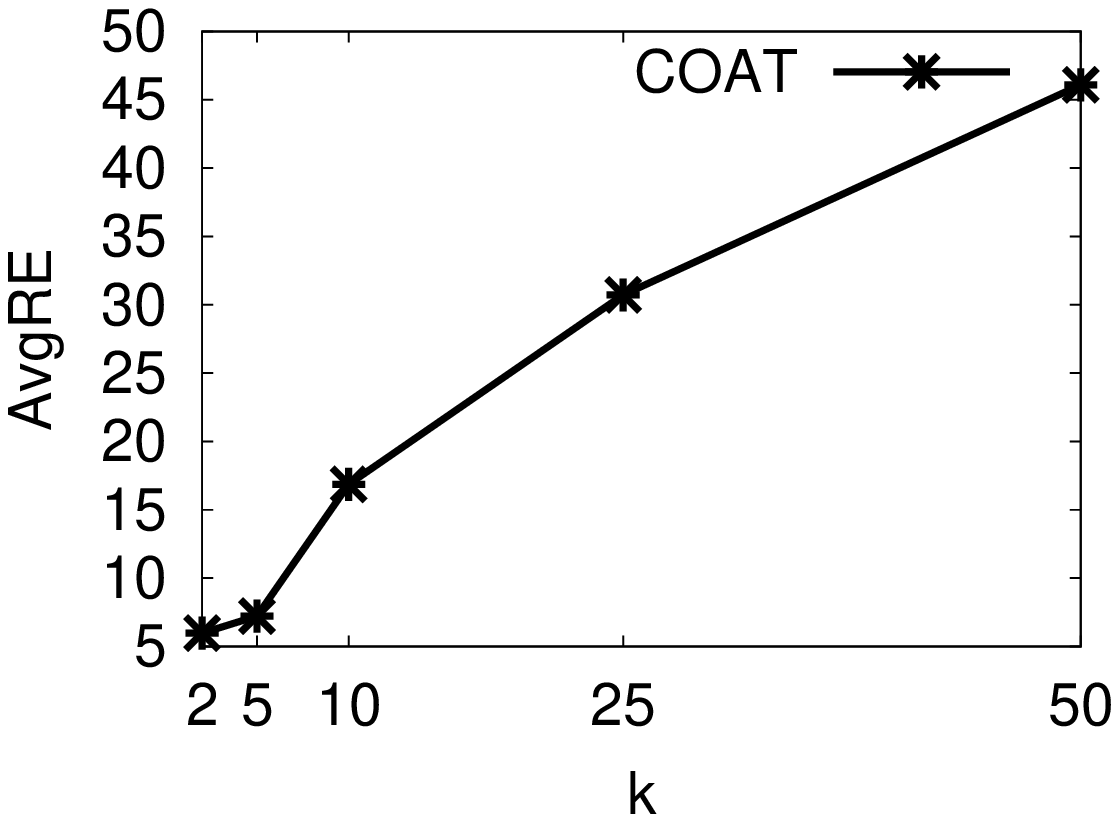}\label{case2}
}
\vspace{-3mm}
\caption{AvgRE vs. $k$ for \emph{EMR-D} computed using (a) \emph{W1}  and (b) \emph{W2}}
\end{figure}

Fig. \ref{case1} reports the \emph{AvgRE} scores for \emph{EMR}, where $k$ was selected over the range $[2,50]$, and \emph{W1} was used. As can be seen,
the \emph{AvgRE} scores indicate that anonymized data permits queries that are common in biomedical data analysis tasks to be answered fairly accurately, even when a strict privacy policy is adopted. The corresponding result for \emph{W2} is reported in Fig. \ref{case2}. Again, the \emph{AvgRE} scores confirm that a low level of information loss was incurred to anonymize \emph{EMR}, particularly when $k$ is $5$ or lower as it is commonly the case when publishing biomedical data \cite{ElEmam}.

In summary, our case study confirms the effectiveness of our anonymization framework when there are specific utility requirements.
This is because it allows the \emph{EMR} dataset to be published in a way that
prevents linking attacks with respect to any portion of any transaction of this dataset, helps
biomedical studies focusing on specific disorders, and allows accurate data analysis.

\section{Conclusions and Future Work}\label{conclusion}

Existing approaches for anonymizing transactional data often produce excessively
distorted data that is of limited utility, due to the fact that they
incorporate coarse privacy requirements, are agnostic with respect to data
utility requirements, and search a fraction of the solution
space. In response, we developed a novel approach that overcomes these limitations by allowing fine-grained privacy and utility requirements to
be specified as constraints, and COAT (COnstraint-based Anonymization of Transactions), an algorithm that transforms data using item generalization and
suppression to satisfy the specified constraints, while minimally distorting data. We also demonstrated the effectiveness of our approach using extensive experiments
on benchmark datasets and a case study on patient-specific data containing diagnosis codes.
Our results demonstrate that COAT is able to satisfy a wide range of privacy and utility requirements with
less information loss than the state-of-the art method, and to anonymize data in a way that prevents identity disclosure and retains data utility for intended applications.

This work also opens up several directions for future investigation.
First, although experimentally shown to be both effective and efficient in practice, the COAT algorithm is heuristic in nature, and, as such, it does not guarantee generating optimal anonymizations in terms of minimum information loss. To address the growing size of datasets and domains, we intend to develop approximation algorithms that can offer such guarantees. Second, we aim at extending our framework to deal with the problem of attribute disclosure based on the $l$-diversity privacy principle \cite{2}.

\bibliographystyle{plain}
\bibliography{bibl}
\end{document}